%% file: QPS-r-PEVA.tex
\documentclass{elsarticle}

\input{libs/common_used_packages}

\usepackage[sort]{cleveref}
\usepackage[group-separator={,},group-minimum-digits=4]{siunitx}

\renewcommand{\hlb}[1]{{\color{black}#1}}
\renewcommand{\hlr}[1]{{\color{black}#1}}
\renewcommand{\hlp}[1]{{\color{black}#1}}
\usepackage{wrapfig}
\usepackage{etoolbox}

\usepackage{IEEEtrantools}
\patchcmd{\IEEEproofindentspace}{2\parindent}{0pt}{}{}

\input{libs/commands-for-response}
\colorlet{red}{black}

\begin{document}
\sloppy

\begin{frontmatter}
\input{sections/title_and_authors}

\input{sections/abstract}

\end{frontmatter}

\input{sections/introduction}

\input{sections/system_model}

\input{sections/scheme}

\input{sections/throughput_delay_analysis}

\input{sections/markovian_arrivals}

\input{sections/evaluation}

\input{sections/related_work}

\input{sections/conclusion}


\bibliographystyle{elsarticle-num}
\bibliography{bibs/QPS,bibs/added_during_revision}

\appendix
\input{sections/appendix}


\end{document}

%% file: libs/common_used_packages.tex
\usepackage{graphicx}
\usepackage{color}
\usepackage{subfigure}
\usepackage{graphics}

\usepackage{enumerate}
\usepackage{tikz}
\usetikzlibrary{shapes,arrows,positioning}

\newcommand{\st}[1]{}

\usepackage{bbm}
\usepackage{amssymb}
\usepackage{amsmath}
\usepackage{amsfonts}
\usepackage{latexsym}
\usepackage{mathtools}
\usepackage{cuted}
\usepackage{verbatim}

\usepackage{textcomp}
\usepackage[latin1]{inputenc}

\usepackage{multirow}

\usepackage{booktabs}

\usepackage{url}

\usepackage{balance}

\usepackage{hyperref}
\usepackage[lined,boxed,linesnumbered,vlined]{algorithm2e}
\usepackage{algpseudocode}

\usepackage{etoolbox}

\def\equationautorefname~#1\null{%
  (#1)\null
}

\SetKwInput{Problem}{problem}
\SetKwInput{Input}{input}
\SetKwInput{Output}{output}
\SetAlgoSkip{smallskip} 
\SetAlgoInsideSkip{smallskip} %
\SetKwComment{tcp}{\#\ }{}
\SetKwComment{tcc}{\#\ }{}
    
\usepackage{amsthm}
\newtheoremstyle{saber}
  {}
  {}
  {\itshape}
  {}
  {\bfseries}
  {.}
  {.5em}
  {}

\theoremstyle{saber}
\ifcsundef{theorem}{
\newtheorem{theorem}{Theorem}
}{}
\ifcsundef{corollary}{%
\newtheorem{corollary}{Corollary}
}{}
\ifcsundef{lemma}{%
\newtheorem{lemma}{Lemma}
}{}
\ifcsundef{definition}{%
\newtheorem{definition}{Definition}
}{}
\ifcsundef{fact}{%
\newtheorem{fact}{Fact}
}{}
\ifcsundef{remark}{%
}{}
\ifcsundef{property}{%
\newtheorem{property}{Property}
}{}

\ifx\theoremautorefname\undefined
\newcommand*{\theoremautorefname}{Theorem}
\fi
\ifx\lemmaautorefname\undefined
\newcommand*{\lemmaautorefname}{Lemma}
\fi
\ifx\definitionautorefname\undefined
\newcommand*{\definitionautorefname}{Definition}
\fi
\ifx\corollaryautorefname\undefined
\newcommand*{\corollaryautorefname}{Corollary}
\fi
\ifx\factautorefname\undefined
\newcommand*{\factautorefname}{Fact}
\fi
\ifx\propertyautorefname\undefined
\newcommand*{\propertyautorefname}{Property}
\fi

\ifx\argmax\undefined
\newcommand{\argmax}{\operatornamewithlimits{argmax}}
\fi
\ifx\argmin\undefined
\newcommand{\argmin}{\operatornamewithlimits{argmin}}
\fi
\ifx\limsup\undefined
\newcommand{\limsup}{\operatornamewithlimits{limsup}}
\fi
\ifx\liminf\undefined
\newcommand{\liminf}{\operatornamewithlimits{liminf}}
\fi
\ifx\norm\undefined
\newcommand{\norm}[1]{\lVert#1\rVert}
\fi
\ifx\abs\undefined
\newcommand{\abs}[1]{\lvert#1\rvert}
\fi
\ifx\set\undefined
\newcommand{\set}[1]{\left\{#1\right\}}
\fi
\ifx\mset\undefined
\newcommand{\mset}[1]{\lbrack #1\rbrack}
\fi
\ifx\etal\undefined
\newcommand{\etal}[1]{{\em #1 et al.}~}
\fi
\ifx\ie\undefined
\newcommand{\ie}{{\em i.e.,} }
\fi
\ifx\eg\undefined
\newcommand{\eg}{{\em e.g.,} }
\fi
\ifx\etc\undefined
\newcommand{\etc}{{\em etc.,} }
\fi
\ifx\wrt\undefined
\newcommand{\wrt}{{\em w.r.t.} }
\fi
\ifx\dotprod\undefined
\newcommand{\dotprod}[2]{
  \langle #1, #2 \rangle
}
\fi
\ifx\iid\undefined 
\newcommand{\iid}{i.i.d.}
\fi
\ifx\bigParenthes\undefined 
\newcommand{\bigParenthes}[1]{
  \big(#1\big)
}
\fi
\ifx\bigBracket\undefined 
\newcommand{\bigBracket}[1]{
  \big\{#1\big\}
}
\fi
\ifx\bigSqBracket\undefined 
\newcommand{\bigSqBracket}[1]{
  \big[#1\big]
}
\fi
\ifx\BigParenthes\undefined 
\newcommand{\BigParenthes}[1]{
  \Big(#1\Big)
}
\fi
\ifx\BigBracket\undefined 
\newcommand{\BigBracket}[1]{
  \Big\{#1\Big\}
}
\fi
\ifx\BigSqBracket\undefined 
\newcommand{\BigSqBracket}[1]{
  \Big[#1\Big]
}
\fi
\ifx\biggParenthes\undefined 
\newcommand{\biggParenthes}[1]{
  \bigg(#1\bigg)
}
\fi
\ifx\biggBracket\undefined 
\newcommand{\biggBracket}[1]{
  \bigg\{#1\bigg\}
}
\fi
\ifx\biggSqBracket\undefined 
\newcommand{\biggSqBracket}[1]{
  \bigg[#1\bigg]
}
\fi
\ifx\BiggParenthes\undefined 
\newcommand{\BiggParenthes}[1]{
  \Bigg(#1\Bigg)
}
\fi
\ifx\BiggBracket\undefined 
\newcommand{\BiggBracket}[1]{
  \Bigg\{#1\Bigg\}
}
\fi
\ifx\BiggSqBracket\undefined 
\newcommand{\BiggSqBracket}[1]{
  \Bigg[#1\Bigg]
}
\fi
\ifx\bracket\undefined
\newcommand{\bracket}[1]{
  \{#1\}
}
\fi
\ifx\parenthes\undefined
\newcommand{\parenthes}[1]{
  (#1)
}
\fi
\ifx\sqBracket\undefined
\newcommand{\sqBracket}[1]{
  [#1]
}
\fi
\ifx\prob\undefined 
\newcommand{\prob}[1]{\mathbb{P}[#1]}
\fi
\ifx\Prob\undefined 
\newcommand{\Prob}[1]{\mathbb{P}\big[#1\big]}
\fi
\ifx\probb\undefined 

\fi
\ifx\expect\undefined 
\newcommand{\expect}[1]{\mathbb{E}[#1]}
\fi
\ifx\Expect\undefined 
\newcommand{\Expect}[1]{\mathbb{E}\big[#1\big]}
\fi
\ifx\expectt\undefined 
\newcommand{\expectt}[1]{\mathbb{E}\bigg[#1\bigg]}
\fi
\ifx\walk\undefined 
\makeatletter
\newcommand{\walk}[1]{%
  \@tempswafalse
  \@for\next:=#1\do
    {\if@tempswa\!\!\rightarrow\!\!\else\@tempswatrue\fi\next}%
}
\makeatother
\fi
\ifx\seq\undefined 
\newcommand{\seq}{\!=\!}
\fi
\ifx\sminus\undefined 
\newcommand{\sminus}{\!-\!}
\fi
\ifx\sm\undefined 
\newcommand{\sm}[1]{\!#1\!}
\fi

\ifx\union\undefined
\newcommand{\union}[2]{#1\!\cup\!#2}
\fi

\ifx\hl\undefined 
\newcommand{\hl}[2]{{\color{#1}#2}}
\fi

\ifx\hlb\undefined 
\newcommand{\hlb}[1]{{\color{blue}#1}}
\fi

\ifx\hlr\undefined 
\newcommand{\hlr}[1]{{\color{red}#1}}
\fi

\ifx\hlp\undefined 
\newcommand{\hlp}[1]{{\color{purple}#1}}
\fi

\ifx\hlc\undefined 
\newcommand{\hlc}[3]{{\color{#1}#2 [remark: #3]}}
\fi

\ifcsdef{sitemize}{}{

}

\ifcsdef{senumerate}{}{

}

\ifcsdef{spenumerate}{}{

}

\ifcsdef{slenumerate}{}{

}


%% file: libs/commands-for-response.tex
\newcounter{rcommentno}[section]

\makeatletter
\newcommand{\rcommentnoautorefname}{Comment \#\@gobble}
\makeatother

%% file: sections/title_and_authors.tex
\title{QPS-r: A Cost-Effective Iterative Switching Algorithm for Input-Queued Switches}
\author[gatechcs]{Long~Gong}
\ead{gonglong@gatech.edu}
\author[gatechcs]{Jun~Xu}
\ead{jx@cc.gatech.edu}
\author[gatechcs]{Liang~Liu}
\ead{liuliang142857@gatech.edu}
\author[gatechisye]{Siva~Theja~Maguluri}
\ead{siva.theja@gatech.edu}

\address[gatechcs]{School of Computer Science, Georgia Institute of Technology}
\address[gatechisye]{School of Industrial \& Systems Engineering, Georgia Institute of Technology}

%% file: sections/abstract.tex
\begin{abstract}
In an input-queued switch, a crossbar schedule,
or a matching between the input ports
and the output ports needs to be computed in each switching cycle, or time slot.
Designing switching algorithms with very low computational complexity, that lead to high throughput and small delay is a challenging problem.
There appears to be a fundamental tradeoff between
the computational complexity of the
switching
algorithm and the resultants throughput and delay. 
Parallel maximal matching algorithms (adapted for switching)
appear to have stricken a sweet spot in this tradeoff, and prior work has shown the following performance guarantees.
%
Using maximal matchings in every time slot results in at least 50\% switch throughput
and {\it order-optimal} ({\it i.e.,} independent of the switch size $N$) average delay bounds  for
various traffic arrival processes.
On the other hand,
their computational complexity can be as low as $O(\log^2 N)$ per port/processor, which is
much lower than those of the algorithms such as maximum
weighted matching which ensures better throughput performance. 

In this work, we propose QPS-r, a parallel iterative
switching algorithm that has the lowest possible computational complexity:  $O(1)$ per port.
Using Lyapunov stability analysis, we show that the
throughput and delay performance is identical to that of maximal matching algorithm. 
Although QPS-r builds upon an existing technique called Queue-Proportional Sampling (QPS), in this paper, we provide analytical guarantees on its throughput and delay under {\it i.i.d.} traffic as well as a Markovian traffic model which can model many realistic traffic patterns. 
We also demonstrate that QPS-3 (running 3 iterations) has comparable empirical throughput and delay performances as iSLIP (running $\log_2 N$ iterations),
a refined and optimized representative maximal matching algorithm adapted for switching.

\end{abstract}



%% file: sections/introduction.tex
\section{Introduction}\label{sec:intro}

The volume of network traffic across the Internet and in data-centers continues to grow relentlessly,
thanks to existing and emerging data-intensive applications, such as big data analytics, cloud computing,
and video streaming.  At the same time, the number of network-connected devices is exploding,
fueled by the wide adoption of
smart phones and the emergence of the Internet of things.
To transport and ``direct" this massive amount of traffic to their respective destinations,
switches and routers capable of
connecting a large number of ports (called {\it high-radix}~\cite{Cakir2015HighRadix,Cakir2016HighRadixCrossbar}) and operating
at very high line rates are badly needed.

Many present day high-performance switching systems in Internet routers and
data-center switches employ an input-queued crossbar to interconnect their input ports and output ports. 
In an $N\times N$ input-queued crossbar switch, each input port can be connected to only one output
port and vice versa in each switching cycle or time slot.
Hence, in every time slot, the switch needs to compute
a one-to-one \textit{matching} between input and output ports ({\it i.e.}, the crossbar schedule).
A major research challenge in designing high-link-rate 
high-radix
switches 
is to develop algorithms that can
compute ``high quality'' matchings -- {\it i.e.,} those that result in
high switch throughput and low queueing
delays for packets -- in a few nanoseconds. 
Clearly, a suitable switching algorithm has to have very low
computational complexity,
yet output ``fairly good'' matching decisions most of time.

\subsection{The Family of Maximal Matchings}

A family of 
parallel iterative algorithms for computing {\it maximal matching}
(one to which no edge can be added for it to remain a matching, a definition that will be made precise in~\autoref{sec:background})
are arguably the best candidates for
switching 
in high-link-rate high-radix switches, because they have reasonably low computational complexities,
yet can provide fairly good 
throughput and delay performance  
guarantees.
More specifically,
using
maximal matchings as crossbar schedules results in at least 50\%
switch throughput in theory (and usually much higher throughput in practice)~\cite{Dai00Speedup}. 
In addition, it results in low packet delays that also have excellent scaling behaviors such as {\it order-optimal} ({\it i.e.,} independent of switch size $N$)
under various traffic arriving processes when the offered load is less than 50\% ({\it i.e.}, within the provable stability region)~\cite{Neely2008DelayMM}. 
In comparison, matchings of higher qualities such as maximum matching (with the largest possible number of edges) and 
maximum weighted matching (with the highest total edge weight) are much more expensive to compute. 
Hence, it is fair to say that, maximal matching algorithms overall
deliver the biggest ``bang'' (performance) for the ``buck'' (computational complexity). 

Unfortunately, parallel maximal matching algorithms are still not ``dirt cheap'' computationally.
More specifically, 
all existing parallel algorithms that compute maximal
matchings on {\it general} $N\times N$ bipartite graphs 
require a minimum of $O(\log N)$ iterations (rounds of message exchanges).
This minimum is attained by the classical algorithm of Israel and Itai~\cite{IsraelItai1986DistMaximalMatching};
the PIM algorithm~\cite{Anderson1993PIM} is a slight adaptation of this classical algorithm
to the switching context, and iSLIP~\cite{McKeown99iSLIP} further improves upon PIM by
reducing its per-iteration per-port computational complexity to $O(\log N)$ via
de-randomizing a computationally expensive ($O(N)$ complexity to be exact) operation in PIM.

\subsection{QPS-r: Bigger Bang for the Buck}

In this work, we propose QPS-r,
a parallel iterative switching algorithm that has the lowest possible computational complexity:  $O(1)$ per port.
More specifically, QPS-r requires
only $r$ (a small constant independent of $N$) iterations to compute a matching, and the computational
complexity of each iteration is only $O(1)$;  here QPS stands for Queue-Proportional Sampling, an add-on
technique proposed in~\cite{Gong2017QPS} that we will describe shortly.
Yet, even the matchings that QPS-1 (running only a single iteration) computes have the same (reasonably high) quality as maximal matchings
in the following sense:
Using such matchings as crossbar schedules results in exactly the same aforementioned provable
throughput and delay guarantees as using maximal matchings, as we will show using Lyapunov stability analysis.
QPS-r performs as well as maximal matching algorithms not just in theory:
We will show 
in~\autoref{sec:evaluation} 
that QPS-3 (running 3 iterations) has comparable empirical throughput and delay performances as iSLIP (running $\log_2 N$ iterations),
a refined and optimized representative maximal matching algorithm adapted for switching, under various workloads. 
Note that matchings that QPS-r computes are generally not maximal. 
QPS-r can make do with less (iterations) because the queue-proportional sampling operation
implicitly makes use of the 
queue (VOQ) length 
information, which maximal matching algorithms do not.  
One major contribution of this work
is to discover the family of (QPS-r)-generated matchings that is even more cost-effective. 


Although QPS-r builds on the QPS data structure and algorithm proposed in~\cite{Gong2017QPS},
our work on QPS-r is very different in three important aspects.  First, in~\cite{Gong2017QPS},
QPS was used only as an add-on to other 
switching 
algorithms such as iSLIP~\cite{McKeown99iSLIP} and SERENA~\cite{GiacconePrabhakarShah2003SERENA} by generating a starter matching for other switching algorithms to further refine,
whereas in this work, QPS-r is used only as a stand-alone algorithm.
%
%
Second, we are the first to discover and prove
that (QPS-r)-generated matchings and maximal matchings provide exactly the same aforementioned 
performance  
guarantees,
whereas in~\cite{Gong2017QPS}, no such mathematical similarity or connection was mentioned.
Third, the establishment of this mathematical similarity is an important theoretical contribution in itself,
because maximal matchings have long been established as a cost-effective family both in switching~\cite{Anderson1993PIM,McKeown99iSLIP} and
in wireless networking~\cite{Neely2008DelayMM,Neely2009MMDelay}, and with this connection we have
considerably enlarged this family.

Although we show that QPS-r has exactly the same throughput and delay bounds as those of maximal
matchings established in~\cite{Dai00Speedup,Neely2008DelayMM,Neely2009MMDelay}, our proofs are 
different 
for the following reason.
A {\it departure inequality} (see~\autoref{lemma:deterministic-depature-inequality}), satisfied by all maximal matching algorithms was used in the 
throughput analysis of~\cite{Dai00Speedup} and the delay
analysis of~\cite{Neely2008DelayMM,Neely2009MMDelay}.
This inequality, however, is not satisfied by QPS-r in general.
Instead, QPS-r satisfies this departure inequality in expectation, which is a much weaker guarantee. 
{\color{red}The methodological contributions of this work are twofold. 
First, we prove two theorems stating that this much weaker guarantee is sufficient for obtaining the same throughput and delay bounds under {\it i.i.d.} traffic arrivals. 
Second, we generalize both results to the more general Markovian arrivals.}

The rest of this paper is organized as follows. In \autoref{sec:background}, we provide some
background on the switching problem in input-queued crossbar switches.  In \autoref{sec:scheme},  we first review
QPS, and then describe QPS-r.
Then in~\Cref{sec:throughput-and-delay-analysis,sec:map}, 
we derive the throughput and
the queue length (and delay) bounds of QPS-r, 
followed by the
performance evaluation in~\autoref{sec:evaluation}. In~\autoref{sec: related-work}, we survey
related work
before concluding this paper in~\autoref{sec:conclusion}.

%% file: sections/system_model.tex
\section{Background on Crossbar Scheduling}\label{sec:background}

In this section, we provide a brief introduction to the crossbar scheduling (switching) problem, and describe and compare the aforementioned three different types of matchings.
Throughout this paper we adopt the standard assumption~\cite[Chapter 2, Page 21]{aweya_switch_router_2018} 
that all the incoming variable-length packets are first segmented into fixed-length
packets (also referred to as cells), and then reassembled at their respective output
ports before leaving the switch. Each fixed-length packet takes one time slot to switch.
We also assume that all input links/ports and output links/ports operate at the same normalized line rate of $1$,
and so do all wires and crosspoints inside the crossbar.

\subsection{Input-Queued Crossbar Switch}\label{subsec:sys-module}

In an $N\times N$ input-queued crossbar switch, 
each input port has $N$ Virtual Output Queues
(VOQs)~\cite{Tamir_HighperformanceMulti_1988}.  The $j^{th}$ VOQ at input port $i$
serves as a buffer for packets
going from input port $i$ to output port $j$. The use of VOQs solves the Head-of-Line
(HOL) blocking issue~\cite{Karol1987VOQHOL}, which
severely limits the throughput of the
switch system.

An $N\times N$ input-queued crossbar can be modeled as a 
bipartite graph, of which the
two disjoint vertex sets are the $N$ input ports and the $N$ output ports
respectively.
In this bipartite graph, there is an
edge between input port $i$ and output port $j$, if and only if
the $j^{th}$ VOQ at input port $i$,
the corresponding VOQ,
is nonempty.  
The weight of this edge is defined
as the length of ({\it i.e.,} the number of packets buffered at) this VOQ.
A set of such edges constitutes a {\it valid crossbar schedule}, or a {\it matching},
if any two of them do not share a common vertex. 
The weight of a matching
is the total weight of all the edges belonging to it ({\it i.e.,} the total length of all corresponding VOQs).

A matching 
$M$ can be represented as an $N\times N$ {\it sub-permutation matrix} (a $0$-$1$ matrix
that contains at most one entry of ``$1$'' in
each row and in each column) $S\!=\!(s_{ij})$ as follows:
$s_{ij}=1$ if and
only if the edge between input port $i$ and
output port $j$ is contained in $M$ ({\it i.e.,} input port $i$ is matched to output port $j$ in $M$).  To avoid any confusion,
only $S$ (not $M$) is used to denote a matching in the sequel, and it can be both a set (of edges) and a matrix.

\subsection{Maximal Matching}\label{subsec:maximal-matching}

As mentioned in~\autoref{sec:intro}, three types of matchings
play important roles in crossbar scheduling problems:
(I) maximal matchings, (II) maximum matchings, and (III) maximum
weighted matchings. A matching $S$ is
called a {\it maximal matching}, if it is no longer a matching,
when any edge not in $S$ is added to it.  A matching
with the largest possible number of edges is called a {\it maximum matching} or {\it maximum cardinality matching}.
Neither maximal matchings nor maximum matchings take into account the
weights of edges, whereas {\it maximum weighted matchings} do.
A maximum weighted matching is one that has the largest total weight among all matchings.
\st{In the switching context (where all edge weights are positive), }
By definition, any maximum matching or maximum weighted matching is
also a maximal matching, but neither converse is generally true.

As mentioned earlier, the family of maximal matchings has long been recognized as a cost-effective family for crossbar scheduling.
Compared to maximal matching, maximum weighted matching (MWM) \hlb{({\it i.e.,} the well-known MaxWeight 
scheduler~\cite{Tassiulas90Max} in 
the context of scheduling transmissions in wireless networking)} is much less cost effective.
\hlb{Although MWM provides stronger provable guarantees such as 
$100$\% switch throughput~\cite{McKeownMekkittikulAnantharamEtAl1999,Tassiulas90MaxWeight} and $O(N)$ 
average packet delay~\cite{Shaf02DelayboundAppMWM} whenever the offered load is less than 100\% in 
theory (and usually much better empirical delay performance in practice as shown 
in~\cite{McKeownMekkittikulAnantharamEtAl1999}),}
the state-of-the-art serial MWM algorithm (suitable for switching) has a prohibitively high 
computational complexity of $O(N^{2.5}\log W)$~\cite{Duan2012MBWP}, where
$W$ is the maximum possible weight (length) of an edge (VOQ).
\st{Even though parallel or distributed MWM algorithms exist ({\it e.g.,}
\mbox{\cite{BayatiShahSharma2008MWMMaxProduct,Fayyazi04ParallelMWM,Fayyazi2006ParaMWM,Goldberg1993sublinearpMWM}}), they
either have a high per-processor complexity ({\it e.g.,} $O(N^2)$ complexity in~\mbox{\cite{BayatiShahSharma2008MWMMaxProduct}} when $N$ processors are used) or require a prohibitively large
number of processors or hardware processing units ({\it e.g.,} $O(N^3)$ or higher
in~\mbox{\cite{Fayyazi04ParallelMWM,Fayyazi2006ParaMWM}}).}  
By the same measure, maximum matching is not a great deal either:  
It is only slightly cheaper
to compute than MWM, yet using maximum matchings as crossbar schedules 
generally cannot \st{provably} guarantee 100\% throughput~\cite{Keslassy_Maximumsizematching_2003}.

%% file: sections/scheme.tex
\section{The QPS-\texorpdfstring{$r$}{r} Algorithm}\label{sec:scheme}

The QPS-r algorithm simply runs $r$ iterations of QPS (Queue-Proportional
Sampling)~\cite{Gong2017QPS} to arrive at a 
matching, so its computational 
complexity 
per port
is exactly $r$ times those of QPS.  
Since $r$ is a small constant, 
it is $O(1)$, same as that of QPS.
In the following two subsections, we describe QPS and QPS-r respectively in more details.


\subsection{Queue-Proportional Sampling (QPS)}\label{subsec:qps}

QPS was used in~\cite{Gong2017QPS} as an ``add-on'' to augment other switching algorithms as follows.  It generates a starter matching, which is then
populated ({\it i.e.}, adding more edges to it) and refined, by other switching algorithms such as iSLIP~\cite{McKeown99iSLIP} and SERENA~\cite{GiacconePrabhakarShah2003SERENA}, into a final matching.
To generate such a starter matching, QPS needs to run only one iteration, 
which
consists of two phases, namely,
a proposing phase and an accepting phase.  
We briefly describe them in this section for this paper to be self-contained.

\smallskip
\noindent
{Proposing Phase.}
In this phase, each input port proposes
to {\it exactly one} output port -- decided by the QPS strategy
-- unless it has no packet to transmit.
Here we will only describe the operations
at input port $1$;  that at any other
input port is identical.
Like in~\cite{Gong2017QPS}, we denote by $m_1,m_2,\cdots,m_N$
the respective queue lengths of the $N$ VOQs at input port $1$,
and by $m$
their total ({\it i.e.}, $m \!\triangleq\!\sum_{k=1}^N m_k$).
Input port 1 simply samples an output
port $j$ with probability $\frac{m_j}{m}$,
{\it i.e.,} proportional to $m_j$, the length of the corresponding VOQ (hence the name QPS);
it then proposes to output port $j$, with the value $m_j$ that will be used in the next
phase.
The computational complexity of this QPS operation, carried out using a simple data structure proposed in~\cite{Gong2017QPS},
is $O(1)$ per (input) port.   



\smallskip
\noindent
{Accepting Phase.}
We describe only the action of output port $1$ in the accepting phase;  
that of any other output port is identical.
The action of output port $1$
depends on the number of proposals it receives.
If it receives exactly one proposal
from an input port, it will accept the proposal
and match with the input port. However, if it receives proposals
from multiple input ports, it will accept the proposal accompanied with the
largest VOQ length (called the ``longest VOQ first" accepting strategy), with ties broken uniformly at random.
The computational complexity of this accepting strategy is $O(1)$ on average and can be made $O(1)$ even in the worst case~\cite{Gong2017QPS}.

\subsection{The QPS-r Scheme}\label{subsec:qps-r}

The QPS-r scheme simply runs $r$ QPS iterations.
In each iteration, each  input port that is not matched yet,  first proposes to an output port according to the QPS proposing strategy;
each output port that is not matched yet, accepts a proposal (if it has received any) according to the ``longest VOQ first'' accepting strategy.
Hence, if an input port has to propose multiple times (once in each iteration), due to all its proposals (except perhaps the last) being rejected,
the identities of the output ports it ``samples'' ({\it i.e.,} proposes to) during these iterations are samples with replacement, which more precisely are {\it i.i.d.} random variables with a queue-proportional distribution.


At the first glance, sampling with replacement may appear to be an obviously suboptimal strategy for the following reason.
There is a nonzero probability for
an input port to propose to the same output port multiple times, but since the first (rejected) proposal implies this output port has already accepted ``someone else'' (a proposal from another input port),
all subsequent proposals to this output port will surely go to waste.
For this reason, sampling without replacement ({\it i.e.}, avoiding all output ports proposed to before) may sound like an obviously better strategy.  However, it is really not, since compared to sampling with replacement,
it has a much higher computational complexity of $O(\log N)$, but improves the throughput and delay performances only slightly according to our simulation studies.

%% file: sections/throughput_delay_analysis.tex
\section{Throughput and Delay Analysis under {\it i.i.d.} Arrivals}\label{sec:throughput-and-delay-analysis}

In this section, we show that QPS-1 ({\it i.e.}, running a single QPS iteration) 
delivers exactly the same provable 
throughput and delay guarantees as maximal matching algorithms under {\it i.i.d.} traffic arrivals. 
When $r\! >\! 1$, QPS-r clearly should have better throughput and 
delay performances than QPS-1, as more input and output ports can be matched up
during subsequent iterations, although we are not able to derive stronger bounds.
 

\input{sections/throughput_delay_analysis_preliminary}

\input{sections/throughput_delay_analysis_qps-r-property}

\input{sections/throughput_delay_analysis_throughput_analysis}

\input{sections/throughput_delay_analysis_delay_analysis_comm}

%% file: sections/throughput_delay_analysis_preliminary.tex
\subsection{Preliminaries}\label{subsec:background-notation}\label{subsec:traffic-model}


In this section, we introduce the notation and assumptions that will later be used in our derivations.
We define three $N\times N$ matrices $Q(t)$, $A(t)$, and $D(t)$.
Let $Q(t) \triangleq \big{(}q_{ij}(t)\big{)}$ be the queue length matrix where each $q_{ij}(t)$
is the length of
the $j^{th}$ VOQ at input port $i$
during time slot $t$.
With a slight abuse of notation, we refer to this VOQ as $q_{ij}$ (without the $t$ term).

\def\wrapwidth{0.55\textwidth}
\begin{wrapfigure}{R}{\wrapwidth}
\centering
\includegraphics[draft=false,width=\wrapwidth]{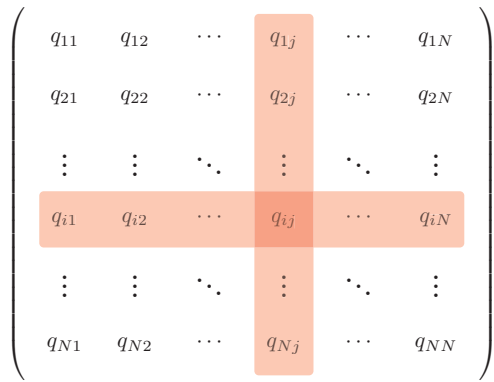}
\caption{$Q^\dagger_{ij}$: neighborhood of $q_{ij}$.}
\label{fig:qd}
\end{wrapfigure}

We define $Q_{i*}(t)$ and $Q_{*j}(t)$ as the sum of the $i^{th}$ row and
the sum of the $j^{th}$
column respectively of $Q(t)$, {\it i.e.}, $Q_{i*}(t)\triangleq\sum_j q_{ij}(t)$ and 
$Q_{*j}(t)\triangleq\sum_i q_{ij}(t)$.
With a similar abuse of notation, we define $Q_{i*}$ as the VOQ set $\{q_{i1}, q_{i2}, \cdots, q_{iN}\}$ ({\it i.e.}, those on
the $i^{th}$ row), and $Q_{*j}$ as $\{q_{1j}, q_{2j}, \cdots, q_{Nj}\}$ ({\it i.e.}, those on the $j^{th}$ column).

Now we introduce a concept that lies at the heart of our derivations: neighborhood.
For each VOQ $q_{ij}$, we define its neighborhood as $Q_{i*}\bigcup Q_{*j}$, the set of VOQs on the
$i^{th}$ row or the $j^{th}$ column.  We denote this neighborhood as $Q^\dagger_{ij}$, since it has the shape of a cross. 
{\color{red}With a slight abuse of notation, we use $Q^\dagger_{ij}$ to also denote the set of the 
corresponding input-output port pairs $\{(l,w):q_{lw}\in Q^\dagger_{ij}\}$.}
\autoref{fig:qd} illustrates $Q^\dagger_{ij}$, where the row and column in the shadow are the VOQ sets $Q_{i*}$ 
and $Q_{*j}$ respectively.  $Q^\dagger_{ij}$ can be viewed as the {\it interference set}   
of VOQs for VOQ $q_{ij}$~\cite{Neely2008DelayMM,Neely2009MMDelay}, as no other VOQ in $Q^\dagger_{ij}$ can be 
active ({\it i.e.,} transmit packets) 
simultaneously with $q_{ij}$. 
We define $Q^\dagger_{ij}(t)$ as the total length of all VOQs in (the set) $Q^\dagger_{ij}$ at time slot $t$, that is
\begin{equation}\label{eq:dagger-def}
Q^\dagger_{ij}(t)\triangleq Q_{i*}(t) - q_{ij}(t) + Q_{*j}(t).
\end{equation}
Here we need to subtract the term $q_{ij}(t)$ so that it is not double-counted (in both $Q_{i*}(t)$ and $Q_{*j}(t))$.

Let $A(t) \!=\! \big{(}a_{ij}(t)\big{)}$ be the traffic arrival matrix where $a_{ij}(t)$
is the number of packets arriving at the input port $i$ destined for
output port $j$ during time slot $t$.  
Here, we assume that, for each $1\!\le\! i, j\!\le\! N$, $\{a_{ij}(t)\}_{t=0}^{\infty}$ is
a sequence of {\it i.i.d.} random variables,
and this sequence is independent of other sequences (for a different $i$ and/or $j$). 
{\color{red}Like in~\cite{mou2020heavy}, we further assume that $a_{ij}(t)$ is upper-bounded by $a_{max}$ for any $i,j$ at any time slot $t$. 
Thus, the second moment
of their common distribution ($=\mathbf{E}\big{[}a^2_{ij}(0)\big{]}$) is finite.} 
Our analysis, however, holds for more general arrival processes ({\it e.g.}, Markovian arrivals) that were 
considered in~\cite{Neely2008DelayMM,Neely2009MMDelay}, as we will elaborate shortly.
Let $D(t) = \big{(}d_{ij}(t)\big{)}$ be the
departure matrix for time slot $t$ output by the 
switching 
algorithm.
Similar to $S$, $D(t)$ is a $0$-$1$ matrix in which
$d_{ij}(t)\!=\!1$ if and only if a packet departs
from $q_{ij}$ during time slot $t$.
For any $i, j$, the queue length process
$q_{ij}(t)$ evolves as follows:
\begin{equation}\label{eq:queueing-dynamics-entrywise}
q_{ij}(t + 1) = q_{ij}(t) - d_{ij}(t) + a_{ij}(t).
\end{equation}

Let $\Lambda = \big{(}\lambda_{ij}\big{)}$ be the (normalized) traffic rate matrix (associated with $A(t)$) where
$\lambda_{ij}$ is normalized (to the percentage of the line rate of an input/output link) mean arrival rate of packets to VOQ $q_{ij}$.  With $a_{ij}(t)$ being an {\it i.i.d.} process, we have $\lambda_{ij} =\mathbf{E}\big{[}a_{ij}(0)\big{]}$.
We define $\rho_\Lambda$ as the maximum load factor imposed on any input or output port by $\Lambda$,
\begin{align}\label{def:lambda_max}
\rho_\Lambda \triangleq \max \big{\{} \max_{1\le i\le N}\{\sum_j \lambda_{ij}\}, \max_{1\le j\le N}\{\sum_i \lambda_{ij}\} \big{\}}
\end{align}
A switching algorithm is said to achieve 100\% throughput or be throughput-optimal if the (packet) queues are stable whenever $\rho_\Lambda < 1$. 

As mentioned before, we will prove 
in this section that, same as the maximal matching algorithms,
QPS-1 is stable {\it under any traffic arrival process $A(t)$}, defined above, whose rate matrix $\Lambda$ satisfies 
$\rho_\Lambda \!<\! 1/2$ ({\it i.e.}, can provably attain at least $50\%$ throughput, or half of the maximum). 
We also derive the average delay bound for QPS-1, which we show is 
{\it order-optimal} ({\it i.e.}, independent of switch size $N$). 
In the sequel, we drop the subscript term from $\rho_\Lambda$ and simply denote it as $\rho$.


Similar to $Q^\dagger_{ij}(t)$, we define $A^\dagger_{ij}(t)$ as the total number of packet arrivals to all VOQs in the neighborhood set $Q^\dagger_{ij}$:
\begin{equation}\label{eq:dagger-def-A}
A^\dagger_{ij}(t)\triangleq A_{i*}(t) - a_{ij}(t) + A_{*j}(t),
\end{equation}
where $A_{i*}(t)$ and $A_{*j}(t)$ are similarly defined as $Q_{i*}(t)$ and $Q_{*j}(t)$ respectively.
$D^\dagger_{ij}(t)$, $D_{i*}(t)$, and $D_{*j}(t)$
are similarly defined, so is $\Lambda^\dagger_{ij}(t)$.
We now state some simple facts concerning $D(t)$, $A(t)$, and $\Lambda$ as follows.
\begin{fact}\label{fact:departure-general}
Given any 
switching 
algorithm, for any $i,j$, we have, $D_{i*}(t)\le 1$ 
(at most one packet can depart from input port $i$ during time slot $t$), $D_{*j}(t)\le 1$, and $D^\dagger_{ij}(t)\le 2$. 

\end{fact}
\begin{fact}\label{fact:arrival-general}
Given any {\it i.i.d.} arrival process $A(t)$ and its rate matrix is $\Lambda$ whose maximum load factor is defined in~\autoref{def:lambda_max}, for any 
$i, j$, we have $\mathbf{E}[A^\dagger_{ij}(t)] = \Lambda^\dagger_{ij} \!\leq\!2\rho$.
\end{fact}
The following fact is slightly less obvious. 
\begin{fact}\label{fact:departure-general-2}
Given any 
switching  
algorithm, for any $i,j$, we have
\begin{equation}\label{eq:d-and-D-dagger}
d_{ij}(t)D^\dagger_{ij}(t)= d_{ij}(t).
\end{equation}
\end{fact}

\autoref{fact:departure-general-2} holds because, as mentioned earlier, 
no other VOQ in $Q^\dagger_{ij}$ (see~\autoref{fig:qd}) can be active simultaneously with $q_{ij}$. 
More precisely, if $d_{ij}(t)\!=\!1$ ({\it i.e.}, 
VOQ $q_{ij}$ is active during time slot $t$) then 
$D^\dagger_{ij}(t)\!\triangleq\!D_{i*}(t)\!-\!d_{ij}(t)\!+\!D_{*j}(t)\!=\!1\!-\!1\!+\!1\!=\!1$; 
otherwise $d_{ij}(t)D^\dagger_{ij}(t)\!=\!0\cdot D^\dagger_{ij}(t)\!=\!0\!=\!d_{ij}(t)$.  





%% file: sections/throughput_delay_analysis_qps-r-property.tex
\subsection{Why QPS-1 Is Just as Good?}\label{subsec:qps-r-property}

The provable throughput and delay bounds of maximal matching algorithms were derived from a ``departure inequality'' (to be stated and proved next) 
that all maximal matchings satisfy.  This inequality, however, is not in general satisfied by matchings generated by QPS-1.   
Rather, QPS-1 satisfies a much weaker form of departure inequality 
(\autoref{lemma:QPS-property}).  
Fortunately, this much weaker condition is sufficient
for proving the same throughput bound and delay bounds, as will be proved in \autoref{thm:stability-tool} and~\autoref{thm:delay-tool} respectively.


\begin{property}[Departure Inequality, stated as Lemma 1 in~\cite{Neely2009MMDelay,Neely2008DelayMM}]\label{lemma:deterministic-depature-inequality}
If during a time slot $t$, the crossbar schedule is a {\it maximal matching}, then each departure process 
$D^{\dagger}_{ij}(t)$ satisfies the following inequality  
\begin{equation}\label{eq:QPS-property-strong}
q_{ij}(t)D^{\dagger}_{ij}(t)\!\geq\! q_{ij}(t).
\end{equation}
\end{property}
\begin{proof}
We reproduce the proof of~\autoref{lemma:deterministic-depature-inequality} with a slightly different approach for 
this paper to be self-contained. 
Suppose the contrary is true, {\it i.e.}, $q_{ij}(t)D^{\dagger}_{ij}(t)\!<\! q_{ij}(t)$.  
This can only happen when 
$q_{ij}(t)\!>\!0$ and $D^{\dagger}_{ij}(t)\!=\!0$. However, $D^{\dagger}_{ij}(t)\!=\!0$ implies that no nonempty VOQ (edge) 
in the neighborhood $Q^{\dagger}_{ij}$ (see \autoref{fig:qd}) is a part of the matching. 
Then this matching cannot be maximal (a contradiction)
since it can be enlarged by the addition of the nonempty VOQ (edge) $q_{ij}$.  
\end{proof}

Clearly, the departure inequality \autoref{eq:QPS-property-strong} above implies the following 
much weaker form of it:
{\color{red}
\begin{equation}\label{eq:QPS-property}
\sum_{i,j}\mathbf{E}\big{[}q_{ij}(t)D^{\dagger}_{ij}(t) \mid Q(t)\big{]} \ge \sum_{i,j}q_{ij}(t)
\end{equation}}
In the rest of this section, we prove the following lemma:
\begin{lemma}\label{lemma:QPS-property}
The matching generated by QPS-1, during any time slot $t$, satisfies the much weaker ``departure inequality''~\autoref{eq:QPS-property}.
\end{lemma}

Before we prove \autoref{lemma:QPS-property}, we introduce an important definition and 
state four facts about QPS-1 that will be used later in the proof.   In the following, 
we will run into several innocuous possible $\frac{0}{0}$ 
situations that all result from queue-proportional sampling, and we consider all of them to be $0$.  

We define $\alpha_{ij}(t)$ as the probability of the event that the proposal from
input port $i$ to output port $j$ is accepted during the accepting phase, conditioned upon the
event that input port $i$ did propose to output port $j$ during the proposing phase.  With this 
definition, we have the first fact
\begin{equation}\label{eq:departure-row-org1}
\mathbf{E}\big{[}d_{ij}(t) \mid Q(t)\big{]}=\frac{q_{ij}(t)}{Q_{i*}(t)}\cdot \alpha_{ij}(t),
\end{equation}
since both sides (note $d_{ij}(t)$ is a 0-1 {\it r.v.}) are the probability that $i$ proposes to $j$ and this 
proposal is accepted.  
Summing over $j$ on both sides, 
we obtain the second fact
\begin{equation}\label{eq:departure-row-org2}
\mathbf{E}\big{[}D_{i*}(t) \mid Q(t)\big{]}=\sum_{j}\frac{q_{ij}(t)}{Q_{i*}(t)}\cdot \alpha_{ij}(t).
\end{equation}


The third fact is that, for any output port $j$,
\begin{equation}\label{eq:departure-column}
\mathbf{E}\big{[}D_{*j}(t) \mid Q(t)\big{]}=1 - \prod_{i}\big{(}1 - \frac{q_{ij}(t)}{Q_{i*}(t)}\big{)}.
\end{equation}
In this equation, the LHS is the conditional {\it probability} ($D_{*j}(t)$ is also a 0-1 {\it r.v.}) that
at least one proposal is received and accepted by output port $j$, and the 
second term on the RHS of~\autoref{eq:departure-column} is the probability that no input port
proposes to output port $j$ (so $j$ receives no proposal).   This equation holds since when $j$ receives one or more proposals, it 
will accept one  
of them (the one with the longest VOQ).




The fourth fact is that, for any $i,j$,
\begin{equation}\label{eq:no-competition}
\alpha_{ij}(t)\geq \prod_{k\neq i}\big{(}1 - \frac{q_{kj}(t)}{Q_{k*}(t)}\big{)}.
\end{equation}
This inequality holds because when input port $i$ proposes to output port $j$, and no other input port does,
$j$ has no choice but to accept $i's$ proposal.  

\subsection{Proof of~\autoref{lemma:QPS-property}}\label{subsubsec:proof-qps-depature}
Now we are ready to prove~\autoref{lemma:QPS-property}.

By the definition of $D^{\dagger}_{ij}(t)\!\triangleq\! D_{i*}(t)\!-\!d_{ij}(t)\!+\!D_{*j}(t)$,
we have,
\begin{align}
\sum_{i,j}\mathbf{E}\big{[}q_{ij}(t)D^{\dagger}_{ij}(t) \mid Q(t)\big{]}
=&{\textstyle\sum\limits_{i,j}}q_{ij}(t)\mathbf{E}\big{[}D_{i*}(t)\mid Q(t)\big{]} \!-\!{\textstyle\sum\limits_{i,j}} q_{ij}(t) \mathbf{E}\big{[}d_{ij}(t)\mid Q(t)\big{]} \nonumber \\
&\qquad {} \!+\! \sum_{i,j} q_{ij}(t)\mathbf{E}\big{[}D_{*j}(t) \mid Q(t)\big{]}.
\label{eq:qps-r-prop-expand}
\end{align}

Focusing on the first term on the RHS of~\autoref{eq:qps-r-prop-expand} and using~\autoref{eq:departure-row-org2}, we have,
\begin{align}
&\sum_{i,j}q_{ij}(t)\mathbf{E}\big{[}D_{i*}(t)\mid Q(t)\big{]}\notag \\
=&{}\sum_{i}Q_{i*}(t)\mathbf{E}\big{[}D_{i*}(t)\mid Q(t)\big{]}\nonumber\\
=&\sum_{i}Q_{i*}(t)\Big{(}\sum_{j}\frac{q_{ij}(t)}{Q_{i*}(t)}\cdot \alpha_{ij}(t)\Big{)} \notag \\
=&\sum_{i,j}q_{ij}(t)\alpha_{ij}(t). \label{eq:qps-r-prop-term-1}
\end{align}

Focusing the second term on the RHS of~\autoref{eq:qps-r-prop-expand} and using
~\autoref{eq:departure-row-org1}, we have
\begin{eqnarray}\label{eq:qps-r-prop-term-3}
-\sum_{i,j} q_{ij}(t) \mathbf{E}\big{[}d_{ij}(t)\mid Q(t)\big{]}\!=\!-\sum_{i,j}q_{ij}(t)\alpha_{ij}(t)\frac{q_{ij}(t)}{Q_{i*}(t)}.
\end{eqnarray}


Hence, the sum of the first two terms in~\autoref{eq:qps-r-prop-expand} is equal to 
\begin{align}
&\sum_{i,j}q_{ij}(t)\alpha_{ij}(t)\Big{(}1\!-\!\frac{q_{ij}(t)}{Q_{i*}(t)}\Big{)}\nonumber\\
\geq&\sum_{i,j}q_{ij}(t)\Big{(}\prod_{k\neq i}\big{(}1 - \frac{q_{kj}(t)}{Q_{k*}(t)}\big{)}\Big{)}\big{(}1 - \frac{q_{ij}(t)}{Q_{i*}(t)}\big{)}\label{eq:cond-depart-explain1}\\
=&\sum_{i,j}q_{ij}(t)\prod_{i}\big{(}1 \!-\! \frac{q_{ij}(t)}{Q_{i*}(t)}\big{)}\nonumber\\
=&\sum_{i,j}q_{ij}(t)\Big{(}1-\mathbf{E}\big{[}D_{*j}(t) \mid Q(t)\big{]}\Big{)}. \label{eq:qps-r-merge-1-2}
\end{align}
Note that \autoref{eq:cond-depart-explain1} is due to \autoref{eq:no-competition} and \autoref{eq:qps-r-merge-1-2} is 
due to \autoref{eq:departure-column}.  We now arrive at~\autoref{eq:QPS-property},
when adding the third and last term in~\autoref{eq:qps-r-prop-expand} to the RHS of~\autoref{eq:qps-r-merge-1-2}.

%% file: sections/throughput_delay_analysis_throughput_analysis.tex
\subsection{Throughput Analysis}\label{subsec:throughput-analysis}
In this section we prove, through Lyapunov stability analysis, 
the following theorem 
({\it i.e.,} \autoref{thm:stability-tool})
which states that any switching algorithm that satisfies the weaker departure inequality~\autoref{eq:QPS-property}, including QPS-1 
as shown in~\autoref{lemma:QPS-property}, can 
attain at least $50\%$ throughput under {\it i.i.d.} arrivals. 
The same throughput bound was proved in~\cite{Dai00Speedup}, through fluid limit analysis, for maximal matching algorithms 
using the (stronger) departure inequality~\autoref{eq:QPS-property-strong}
which as stated earlier is not in general satisfied by matchings generated by QPS-1. 


\begin{theorem}\label{thm:stability-tool}
Let $\{Q(t)\}_{t=0}^{\infty}$ be the queueing process of a switching system that is an irreducible Markov chain.
Let the departure process of $\{Q(t)\}_{t=0}^{\infty}$ satisfy the weaker ``departure inequality''~\autoref{eq:QPS-property}.
Then whenever its maximum load factor $\rho<1/2$, 
the queueing process is stable in the following sense:  (I)  
The Markov chain $\{Q(t)\}_{t=0}^{\infty}$ is positive recurrent
and hence converges to a stationary distribution $\bar{Q}$;  
(II) The first moment of $\bar{Q}$ is finite.  
\end{theorem}

\begin{proof}
Here we prove only (I), since~\autoref{thm:delay-tool} that we will shortly prove implies (II).
We define the following Lyapunov function of $Q(t)$:
$L\big{(}Q(t)\big{)}=\sum_{i,j}q_{ij}(t)Q^\dagger_{ij}(t)$,
where $Q^\dagger_{ij}(t)$ is defined earlier in~\autoref{eq:dagger-def}.  This Lyapunov function was first introduced in~\cite{Neely2008DelayMM} for
the delay analysis of maximal matching algorithms for wireless networking. 
By the Foster-Lyapunov stability criterion~\cite[Proposition 2.1.1]{brucelecture2006}, 
to prove that $\{Q(t)\}_{t=0}^{\infty}$ is positive recurrent, it suffices to show that, there exists a constant $B \!>\! 0$ such that whenever the total queue (VOQ) length $\|Q(t)\|_1 \!>\! B$
(because it is not hard to verify that the complement set of states $\{Q(t): \|Q(t)\|_1 \!\le\! B\}$ is finite and the drift is bounded whenever $Q(t)$ belongs to this set), we have
\begin{equation}\label{eq:qps-drift}
\mathbf{E}\big{[}L\big{(}Q(t+1)\big{)} - L\big{(}Q(t)\big{)} \mid Q(t)\big{]} \leq -\epsilon,
\end{equation}
where $\epsilon > 0$ is a constant.  
It is not hard to check (for more detailed derivations, please refer to~\cite{Neely2008DelayMM}),
\begin{align}
L\big{(}Q(t+1)\big{)}-L\big{(}Q(t)\big{)} 
=&2\sum_{i,j}q_{ij}(t)\big{(}A^\dagger_{ij}(t)-D^\dagger_{ij}(t)\big{)}  \nonumber\\
&\qquad {} +\sum_{i,j}\big{(}a_{ij}(t)-d_{ij}(t)\big{)}\big{(}A^\dagger_{ij}(t)-D^\dagger_{ij}(t)\big{)}. \label{eq:uncond-drift}
\end{align}

Hence the drift (LHS of~\autoref{eq:qps-drift}) can be written as
\begin{align}
&\mathbf{E}\big{[}L\big{(}Q(t+1)\big{)} - L\big{(}Q(t)\big{)} \mid Q(t)\big{]}\nonumber\\
=&\mathbf{E}\Big{[}2\sum_{i,j}q_{ij}(t)\big{(}A^\dagger_{ij}(t)-D^\dagger_{ij}(t)\big{)}\mid Q(t)\Big{]} \nonumber\\
&\qquad{}+\mathbf{E}\Big{[}\sum_{i,j}\big{(}a_{ij}(t)\!-\!d_{ij}(t)\big{)}\big{(}A^\dagger_{ij}(t)\!-\!D^\dagger_{ij}(t)\big{)}\mid Q(t)\Big{]}. \label{eq:raw-drift}
\end{align}



Now we claim the following two inequalities, which we will prove shortly.
\begin{eqnarray}
\mathbf{E}\big{[}2{\textstyle \sum\limits_{i,j}}q_{ij}(t)\big{(}A^\dagger_{ij}(t)\!-\!D^\dagger_{ij}(t)\big{)}\mid Q(t)\big{]}\!\leq\! 2(2\rho \!-\! 1)\|Q(t)\|_1. \label{eq:term-1-tbp}\\
\mathbf{E}\big{[}{\textstyle \sum\limits_{i,j}}\big{(}a_{ij}(t)\!-\!d_{ij}(t)\big{)}\big{(}A^\dagger_{ij}(t)\!-\!D^\dagger_{ij}(t)\big{)}\mid Q(t)\big{]} \!\leq\! CN^2. \label{eq:term-2-tbp}
\end{eqnarray}

With~\autoref{eq:term-1-tbp} and~\autoref{eq:term-2-tbp} substituted into~\autoref{eq:raw-drift}, we have 
\[
\mathbf{E}\big{[}L\big{(}Q(t+1)\big{)} - L\big{(}Q(t)\big{)} \mid Q(t)\big{]} \!\leq\! 2(2\rho \!-\! 1)\|Q(t)\|_1 + CN^2.
\]  
where $C > 0$ is a constant.
Since $\rho \!<\! 1/2$, we have $2\rho-1\!<\!0$. Hence, there exist $B,\epsilon >0$ such that, whenever $\|Q(t)\|_1 \!>\! B$,
\begin{equation*}
\mathbf{E}\big{[}L\big{(}Q(t+1)\big{)} - L\big{(}Q(t)\big{)} \mid Q(t)\big{]} \leq -\epsilon.
\end{equation*}

Now we proceed to prove~\autoref{eq:term-1-tbp}.
\begin{align}
&\mathbf{E}\Big{[}2\sum_{i,j}q_{ij}(t)\big{(}A^\dagger_{ij}(t)-D^\dagger_{ij}(t)\big{)}\mid Q(t)\Big{]}\nonumber\\
=&2\Big{(}\sum_{i,j}\mathbf{E}\Big{[}q_{ij}(t)A^\dagger_{ij}(t)\mid Q(t) \Big{]}\!-\!\sum_{i,j}\mathbf{E}\Big{[}q_{ij}(t)D^\dagger_{ij}(t)\mid Q(t)\Big{]}\Big{)}\nonumber\\
\leq&2\Big{(}2\rho \sum_{i,j}\mathbf{E}\Big{[}q_{ij}(t) \mid Q(t)\Big{]}\!-\!\sum_{i,j}q_{ij}(t)\Big{)} \label{eq:term-1-explain-01}\\
=&2(2\rho-1)\|Q(t)\|_1. \label{eq:bounded-first-term}
\end{align}
In the above derivations, inequality~\autoref{eq:term-1-explain-01} holds due to~\autoref{eq:QPS-property}, $A(t)$ being independent of $Q(t)$ for any $t$, and 
\autoref{fact:arrival-general} that $\mathbf{E}[A^\dagger_{ij}(t)]\!\leq\!2\rho$.


Now we proceed to prove~\autoref{eq:term-2-tbp}, which upper-bounds the conditional expectation $\mathbf{E}\big{[}\big{(}a_{ij}(t)\!-\!d_{ij}(t)\big{)}\big{(}A^\dagger_{ij}(t)\!-\!D^\dagger_{ij}(t)\big{)} \mid Q(t)\big{]}$.
It suffices however to upper-bound the unconditional expectation $\mathbf{E}\big{[}\big{(}a_{ij}(t)\!-\!d_{ij}(t)\big{)}\big{(}A^\dagger_{ij}(t)\!-\!D^\dagger_{ij}(t)\big{)}\big{]}$, which we will do in the following, since we can obtain the same upper bounds on $\mathbf{E}[D^\dagger_{ij}(t)]$ and $\mathbf{E}[d_{ij}(t)]$ ($2$ and $1$ respectively) 
whether the expectations are conditional (on $Q(t)$) or not.  Note the other two terms $A^\dagger_{ij}(t)$ and $a_{ij}(t)$ are independent of (the condition) $Q(t)$.

As for any $i,j$, $a_{ij}(t)$ is {\it i.i.d.}, we have,
\begin{align}
&\mathbf{E}\big{[}\big{(}a_{ij}(t)-d_{ij}(t)\big{)}\big{(}A^\dagger_{ij}(t)-D^\dagger_{ij}(t)\big{)}\big{]}\label{eq:term-2-uncond}\\
=&\mathbf{E}[a_{ij}(t)A^\dagger_{ij}(t) - d_{ij}(t)A^\dagger_{ij}(t)-a_{ij}(t)D^\dagger_{ij}(t)+d_{ij}(t)D^\dagger_{ij}(t)]\nonumber\\
=&\mathbf{E}[a^2_{ij}(t)]\!-\!\lambda^2_{ij}\!+\!\lambda_{ij}\Lambda^{\dagger}_{ij}
\!-\!\mathbf{E}[d_{ij}(t)]\Lambda^\dagger_{ij} 
-\lambda_{ij}\mathbf{E}[D^\dagger_{ij}(t)]\!+\! \mathbf{E}[d_{ij}(t)D^\dagger_{ij}(t)]\nonumber\\
=&\mathbf{E}[a^2_{ij}(t)]\!-\!\lambda^2_{ij}\!+\!\lambda_{ij}\Lambda^{\dagger}_{ij}
\!-\!\mathbf{E}[d_{ij}(t)]\Lambda^\dagger_{ij} 
-\lambda_{ij}\mathbf{E}[D^\dagger_{ij}(t)]\!+\! \mathbf{E}[d_{ij}(t)]. \label{eq:bt}
\end{align}

In arriving at~\autoref{eq:bt}, we have used~\autoref{eq:d-and-D-dagger}.
The RHS of~\autoref{eq:bt} can be bounded by a constant $C > 0$ due to the following assumptions and facts: 
$\mathbf{E}[a^2_{ij}(t)]\!=\!\mathbf{E}[a^2_{ij}(0)]\!<\!\infty$ for any $t$, 
$d_{ij}(t)\!\leq\! 1$, $D^\dagger_{ij}(t)\!\leq\!2$, $\lambda_{ij}\!\leq\!\rho\!<\!1/2$, and $\Lambda^\dagger_{ij}\!\leq\! 2\rho\!<\!1$.
Therefore, we have (by applying $\sum_{i,j}$ to both~\autoref{eq:term-2-uncond} and the RHS of~\autoref{eq:bt})
\begin{equation*}\label{eq:bounded-term}
\sum_{i,j}\mathbf{E}\Big{[}\big{(}a_{ij}(t)\!-\!d_{ij}(t)\big{)}\big{(}A^\dagger_{ij}(t)\!-\!D^\dagger_{ij}(t)\big{)}\Big{]}\!\leq\! CN^2.
\end{equation*}
\end{proof}

\noindent
{\bf Remarks.}  Now that we have proved that $\{Q(t)\}_{t=0}^{\infty}$ is positive recurrent. 
Hence, we have, in steady state, for any $1\le i,j\le N$, $\mathbf{E}[d_{ij}(t)]\!=\!\lambda_{ij}$.  
%
Therefore, we have, in steady state, for any $1\le i,j\le N$
\begin{align}\label{eq:for-delay-analysis}
\mathbf{E}\big{[}\big{(}a_{ij}(t)\!-\!d_{ij}(t)\big{)}\big{(}A^\dagger_{ij}(t)\!-\!D^\dagger_{ij}(t)\big{)}\big{]}
= \sigma^2_{ij}\!-\!\lambda_{ij}\Lambda^\dagger_{ij}\!+\!\lambda_{ij}.
\end{align}
where $\sigma^2_{ij}\!=\!\mathbf{E}[a^2_{ij}(t)]\!-\!\lambda^2_{ij}$ is the variance of 
$a_{ij}(t)$, 
because \autoref{eq:bt} can be simplified as the RHS of~\autoref{eq:for-delay-analysis} in steady state.



Now, we prove the following corollary of~\autoref{thm:stability-tool}
which, in combination with~\autoref{lemma:QPS-property}, shows that QPS-1 can attain at least 50\% throughput. 
\begin{corollary}\label{coro:stability-qps}
Under an i.i.d. arrival process, whenever the maximum load factor $\rho\!<\!1/2$,
QPS-1 is stable in the following sense:
The resulting queueing process $\{Q(t)\}_{t=0}^{\infty}$
is a positive recurrent Markov chain 
and its stationary distribution $\bar{Q}$ has finite first moment. 
\end{corollary}
\begin{proof}
$\{Q(t)\}_{t=0}^{\infty}$ is clearly a Markov chain, since in~\autoref{eq:queueing-dynamics-entrywise}, 
the term $d_{ij}(t)$ is a function of 
$Q(t)$ and $a_{ij}(t)$ is a random variable independent of $Q(t)$. 
{\color{red}Its irreducibility is proved in~\ref{app:proof-markov-for-qps-1}.} 
The rest follows from~\autoref{lemma:QPS-property} and 
\autoref{thm:stability-tool}.
\end{proof}



%% file: sections/throughput_delay_analysis_delay_analysis_comm.tex
\subsection{Delay Analysis}\label{subsec: delay-analysis}

In this section, we derive the bound on the expected total queue length $\mathbf{E}[\|\bar{Q}\|_1]$ (readily convertible to the corresponding delay
bound using Little's Law) for QPS-1 under {\it i.i.d.} traffic arrivals using
the following moment bound lemma ({\it i.e.,} \autoref{thm:moment-bound})~\cite[Proposition 2.1.4]{brucelecture2006}.
This bound, shown in~\autoref{eq:final-delay-bound}, is identical to that derived in~\cite[Section III.B]{Neely2008DelayMM,Neely2009MMDelay} for maximal matchings
under {\it i.i.d.} traffic arrivals.  Note this equivalence is not limited to {\it i.i.d.} traffic arrivals:
As will be shown in~\autoref{subsec:map-delay-ana}, the 
delay analysis results for general Markovian arrivals~\cite{mou2020heavy} derived in~\cite{Neely2008DelayMM,Neely2009MMDelay} for maximal
matchings (using the stronger ``departure inequality''~\autoref{eq:QPS-property-strong}) hold also for QPS-1.

\begin{lemma}\label{thm:moment-bound}
Suppose that $\{Y_t\}_{t=0}^{\infty}$ is a positive recurrent Markov chain with countable state space $\mathcal{Y}$. Suppose $V$, $f$, and $g$ are non-negative functions on $\mathcal{Y}$ such that,
\begin{equation}\label{eq:mbreq}
V(Y_{t+1}) -V(Y_{t})\leq -f(Y_{t}) + g(Y_{t}), \,\mbox{for all } Y_{t} \in \mathcal{Y}.
\end{equation}
Then $\mathbf{E}[f(\bar{Y})]\!\leq\! \mathbf{E}[g(\bar{Y})]$, where $\bar{Y}$ is a random variable with the stationary 
distribution of the Markov chain $\{Y_t\}_{t=0}^\infty$.
\end{lemma}


Now we derive the following bound on $\mathbf{E}[\|\bar{Q}\|_1]$, which is stronger than the part (II) of~\autoref{thm:stability-tool}.
\begin{theorem}\label{thm:delay-tool}
Under the same assumptions and definitions as in~\autoref{thm:stability-tool}, we have
\begin{equation}\label{eq:final-delay-bound}
\mathbf{E}[\|\bar{Q}\|_1]\leq\frac{1}{2(1-2\rho)}\sum_{i,j}\big{(}\sigma^2_{ij}\!-\!\lambda_{ij}\Lambda^\dagger_{ij}\!+\!\lambda_{ij}\big{)}.
\end{equation}
\end{theorem}

\begin{proof}

{\color{red}
We replace function $V$ in~\autoref{thm:moment-bound} by $L$, the Lyapunov function used in the proof of~\autoref{thm:stability-tool} 
and $Y_t$ by the queue length matrix $Q(t)$. Let $f(Y_t)\triangleq -2\sum_{i,j}q_{ij}(t)\big{(}A^\dagger_{ij}(t)-D^\dagger_{ij}(t)\big{)} + h(Y_t)$, where $h(Y_t)\triangleq 4Na_{max}\|Q(t)\|_1$, and $g(Y_t)\triangleq \sum_{i,j}\big{(}a_{ij}(t)\!-\!d_{ij}(t)\big{)}\big{(}A^\dagger_{ij}(t)\!-\!D^\dagger_{ij}(t)\big{)} + h(Y_t)$. It is not hard to check that both $f(\cdot)$ and $g(\cdot)$ 
are non-negative functions and inequality~\autoref{eq:mbreq} holds for $V,Y_t,f,g$ defined above. 
Furthermore, we have proved before, $\{Q(t)\}_{t=0}^{\infty}$ is a positive Markov chain whenever the 
maximum load factor $\rho<1/2$. Hence, we have, in steady state, 
\begin{align}
&-2(2\rho-1)\mathbf{E}[\|\bar{Q}\|_1]\nonumber\\
=&-2(2\rho-1)\mathbf{E}[\|Q(t)\|_1]\nonumber\\
\le&\mathbf{E}\big{[}-2{\textstyle \sum\limits_{i,j}}q_{ij}(t)\big{(}A^\dagger_{ij}(t)\!-\!D^\dagger_{ij}(t)\big{)}\big{]} \label{eq:qb-explain-01}\\
=& \mathbf{E}[f(Y_t) - h(Y_t)] \notag \\
=& \mathbf{E}[f(\bar{Y}) - h(\bar{Y})] \notag \\
=&\mathbf{E}[f(\bar{Y})] - \mathbf{E}[h(\bar{Y})] \notag\\
\leq&\mathbf{E}[g(\bar{Y})]- \mathbf{E}[h(\bar{Y})]\label{eq:qb-explain-02}\\
=&\mathbf{E}[g(\bar{Y})- h(\bar{Y})]\notag\\
=&\mathbf{E}\big{[}\big{(}a_{ij}(t)\!-\!d_{ij}(t)\big{)}\big{(}A^\dagger_{ij}(t)\!-\!D^\dagger_{ij}(t)\big{)}\big{]} \notag \\
=&\sum_{i,j}\big{(}\sigma^2_{ij}\!-\!\lambda_{ij}\Lambda^\dagger_{ij} +\lambda_{ij}\big{)}.\label{eq:qb-explain-03}
\end{align}
In the above derivation, inequality \autoref{eq:qb-explain-01} is
because taking expectation on both sides of~\autoref{eq:term-1-tbp}, we have $\mathbf{E}\big{[}2{\textstyle \sum\limits_{i,j}}q_{ij}(t)\big{(}A^\dagger_{ij}(t)\!-\!D^\dagger_{ij}(t)\big{)}\big{]}\le 2(2\rho-1)\mathbf{E}[\|Q(t)\|_1]$. 
Thus, we obtain
\[
-2(2\rho-1)\mathbf{E}[\|Q(t)\|_1] \le \mathbf{E}\big{[}-2{\textstyle \sum\limits_{i,j}}q_{ij}(t)\big{(}A^\dagger_{ij}(t)\!-\!D^\dagger_{ij}(t)\big{)}\big{]}.
\] 
Inequality \autoref{eq:qb-explain-02} is
due to~\autoref{thm:moment-bound}, and equality~\autoref{eq:qb-explain-03} is
due to~\autoref{eq:for-delay-analysis}.
}

Therefore, we have, in steady state, 
\begin{equation*}
\mathbf{E}[\|\bar{Q}\|_1]\leq\frac{1}{2(1-2\rho)}\sum_{i,j}\big{(}\sigma^2_{ij}\!-\!\lambda_{ij}\Lambda^\dagger_{ij}\!+\!\lambda_{ij}\big{)}.
\end{equation*}
\end{proof}

Since, as explained in the proof of~\autoref{coro:stability-qps}, 
$\{Q(t)\}_{t=0}^{\infty}$ is an irreducible 
Markov chain under {\it i.i.d.} arrivals whose maximum load factor $\rho<1/2$,~\autoref{thm:delay-tool} applies to QPS-1. Hence we obtain,
\begin{corollary}\label{coro:qps-delay}
The bound on $\mathbf{E}[\|\bar{Q}\|_1]$ as stated in~\autoref{eq:final-delay-bound} holds under QPS-1 scheduling, 
whenever the arrival process is {\it i.i.d.} and the maximum load factor $\rho<1/2$. 
\end{corollary}

It is not hard to check (by applying Little's Law) that the average delay (experienced by packets) is
bounded by a constant independent of $N$ ({\it i.e., order-optimal}) for a given maximum load factor $\rho<1/2$,
since the variance $\sigma^2_{ij}$ for any $i,j$ is assumed to be finite in~\autoref{subsec:background-notation}. 
For the special case of Bernoulli {\it i.i.d.} arrival (when $\sigma^2_{ij}=\lambda_{ij}-\lambda^2_{ij}$),
this bound (the RHS) can be further tightened to $\frac{\sum_{i,j}\lambda_{ij}}{1-2\rho}$.  This implies, by Little's Law, the following ``clean'' bound: $\bar{\omega}\leq \frac{1}{1-2\rho}$
where $\bar{\omega}$ is 
the expected delay averaged over
all packets transmitting through the switch.

%% file: sections/markovian_arrivals.tex
\section{Throughput and Delay Analysis under Markovian Arrivals}\label{sec:map}

In this section, we will generalize our results  in~\autoref{sec:throughput-and-delay-analysis} 
to the more general Markovian arrivals~\cite{mou2020heavy}. 

\subsection{Preliminaries}\label{subsec:prelim-map}

As mentioned earlier the traffic arrival matrix $A(t)$ now are Markovian arrivals. 
More precisely, for any $1\le i,j\le N$, the number of arrivals $a_{ij}(t)\triangleq \eta(x_{ij}(t))$
for some non-negative integer-valued function $\eta(\cdot)$, where $\{x_{ij}(t)\}_{t=0}^{\infty}$ is
an irreducible, positive recurrent and aperiodic discrete-time Markov chain (DTMC) on a finite-state space $\mathcal{X}_{ij}\triangleq\{1,2,\cdots,\chi_{ij}\}$ for some integer $\chi_{ij}>0$ 
and $\{X(t)\}_{t=0}^{\infty}$ is also an irreducible, positive recurrent and aperiodic DTMC where $X(t)\triangleq \big{(}x_{ij}(t)\big{)}$. 

Define the steady-state arrival rate $\lambda_{ij}\triangleq \mathbf{E}[\eta(\bar{x}_{ij})]$, where
$\bar{x}_{ij}$ is the stationary distribution that $\{x_{ij}(t)\}_{t=0}^{\infty}$ converges to.
Like in~\cite{Neely2008DelayMM,Neely2009MMDelay}, we assume that the underlying DTMC 
$\{x_{ij}(t)\}_{t=0}^\infty$ for all $1\le i,j\le N$ is in steady state at time
$0$, thus each arrival process $\{a_{ij}(t)\}_{t=0}^{\infty}$ is stationary, and therefore, we have,
$\mathbf{E}[a_{ij}(t)]=\mathbf{E}[\eta(\bar{x}_{ij})]=\lambda_{ij}$ for any $1\le i,j\le N$ and any
time slot $t\ge 0$. Hence, \autoref{fact:arrival-general} also holds for the arrival process $A(t)$. 
Like in~\autoref{subsec:background-notation}, 
we assume that $a_{ij}(t)$ is upper-bounded by $a_{max}$ for any $i,j$ at any time slot $t$. 
We further assume that given any integer $\tau >0$, there exists some $\beta >0$ such that $\mathbf{P}[A(t')=\mathbf{0}]\ge \beta$ 
for any $t' \in \{t,t+1,\cdots,t+\tau\}$, where $\mathbf{0}$ is the $N\times N$ matrix with all its entries equal to 0. 
In other words, the switch could, with a nonzero
probability, have no packet arrivals to any of its VOQs during time slots $t,t+1,\cdots,t+\tau$. 

As mentioned before, we will prove in this section that,
QPS-1 has the same provable throughput and delay guarantees as using maximal matchings
under Markovian arrival processes $A(t)$. Before that, we
state an fact concerning Markovian arrival processes $A(t)$.

\begin{fact}\label{fact:dependent-arrival-rate-map}
Given a Markovian arrival process $\{a_{ij}(t)\}_{t=0}^{\infty}$ defined above, 
let $H(t)$ be the history of all arrivals (for all VOQs) up to but excluding time slot $t$, 
then regardless of the past history $H(t)$, we have, 
\begin{equation}\label{eq:arrival-exp-general}
    \mathbf{E}\left[a_{ij}(t)\mid H(t-T)\right]\leq \lambda_{ij}+\zeta_{ij}(T),
\end{equation}
where $\zeta_{ij}(T)$ is some non-negative function of integer $T$ such that it converges to
$0$ exponentially fast as $T$ approaches to $\infty$.
\end{fact}

The proof of the above fact can be found in Appendix A.2 of~\cite{mou2020heavy}.

For any non-negative integer $T$, we define
\begin{equation}\label{eq:load-bound}
   \xi(T)\triangleq \max_{l,w}\left[\Lambda_{lw}^\dagger +\sum_{(i,j)\in Q^\dagger_{lw}}\zeta_{ij}(T)\right].
\end{equation}

For any $\rho < 1/2$, we define 
\begin{equation}\label{fact:maximum-load-bounded}
K_\rho\triangleq \min\{k\in \mathbb{N}_+: \xi(T) < 1 \mbox{ for } T\ge k \}. 
\end{equation}
Note that such a $K_\rho$ exists because  
$\sum_{(i,j)\in Q^\dagger_{lw}}\zeta_{ij}(T) \rightarrow 0$ as $T\rightarrow \infty$ for any $l,w$ and 
when $\rho <1/2$, we have $\Lambda_{lw}^\dagger < 1$ for any $l,w$.


\subsection{Throughput Analysis}\label{subsec:map-throughput-delay-ana}

We now present the throughput result under Markovian arrivals, which generalizes the throughput result under {\it i.i.d.} arrivals, {\it i.e.}, \autoref{thm:stability-tool}. 
\begin{theorem}\label{thm:stability-tool-map}
Let $A(t)$ be Markovian arrivals determined by the underlying DTMC $\{X(t)\}_{t=0}^{\infty}$ as 
described before in \autoref{subsec:prelim-map}. 
Suppose the joint process $\{(Q(t),X(t))\}_{t=0}^{\infty}$ under Markovian arrivals 
$A(t)$ is an irreducible Markov chain. Let the departure process of $\{Q(t)\}_{t=0}^{\infty}$ satisfies the weaker ``departure inequality''~\autoref{eq:QPS-property}.
Then whenever the maximum load factor $\rho<1/2$, the switching system is stable in the following sense:
(I) The Markov chain $\{\left(Q(t),X(t)\right)\}_{t=0}^{\infty}$ is positive recurrent and hence
 the queueing process $\{Q(t)\}_{t=0}^{\infty}$ converges in distribution to a stationary distribution $\bar{Q}$; (II)
 The first moment of $\bar{Q}$ is finite.
\end{theorem}
\begin{proof}
See \ref{app:proof-stability}.
\end{proof}

It is not hard to check that \autoref{thm:stability-tool-map} applies to 
QPS-1. The proof can be found in~\ref{app:proof-markov-map-for-qps-1}. 
Therefore, QPS-1 can also attain at least 50\% throughput under Markovian arrivals.

\subsection{Delay Analysis}\label{subsec:map-delay-ana}

Now, we present bounds on the expected total queue length $\mathbf{E}[\|\bar{Q}\|_1]$ for QPS-1 under Markovian arrivals.

\begin{theorem}\label{thm:delay-tool-map}
Under the same assumptions and definitions as in~\autoref{thm:stability-tool-map},
we have: (I) The average (total) queue length is upper-bounded as follows.
\begin{align}
\mathbf{E}[\|\bar{Q}\|_1]&\leq\frac{1}{2(1-\xi(K_\rho))}\left(\sum_{i,j}\left(\mathbf{E}[a_{ij}(0)A^\dagger_{ij}(0)] +\lambda_{ij}\right) \right.\notag \\
&\qquad{}\left.+ 2\sum_{i,j}\mathbf{E}\left[\sum_{k=0}^{K_\rho-1}a_{ij}(k-K_\rho)A^\dagger_{ij}(0)\right]\right), \label{eq:final-delay-bound-general}
\end{align}
where $\xi(\cdot)$ and $K_\rho$ are defined in~\autoref{eq:load-bound} and~\autoref{fact:maximum-load-bounded} respectively.

(II) If arrival processes $a_{ij}(t)$ for different $i,j$ are independent of
each other, then the above queue length bound can be further simplified as follows.
\begin{align}
\mathbf{E}[\|\bar{Q}\|_1]\!\le\!\frac{1}{2(1\!-\!\xi(K_\rho))}\left(\sum_{i,j}\left(\sigma_{ij}^2 \!+\! \lambda_{ij}\Lambda_{ij}^\dagger+\lambda_{ij}\right) \!+\! 2\sum_{i,j}\sum_{k=1}^{K_\rho}\left(\lambda_{ij}\Lambda_{ij}^\dagger \!+\!\theta_{ij}(k)\right)\right), \label{eq:final-delay-bound-general-independent}
\end{align}
where $\sigma^2_{ij}\triangleq \mathbf{E}\left[a_{ij}^2(t)\right] - \lambda^2_{ij}$ is the steady-sate variance of $a_{ij}(t)$ and 
$\theta_{ij}(k)\triangleq \mathbf{E}\left[a_{ij}(k+t)a_{ij}(t)\right] - \lambda^2_{ij}$ is
the auto-correlation in $a_{ij}(t)$ in steady state.
\end{theorem}

\begin{proof}
See \ref{app:proof-delay}.
\end{proof}

Similarly, we have that the bounds on $\mathbf{E}[\|\bar{Q}\|_1]$ as stated 
in~\autoref{thm:delay-tool-map} hold under QPS-1 scheduling,
whenever the arrival processes are Markovian arrivals and the maximum load factor $\rho<1/2$.


%% file: sections/evaluation.tex
\section{Evaluation}\label{sec:evaluation}

In this section, we evaluate, through simulations, the performance of
QPS-r under various load conditions and traffic patterns. 
The main purpose of this section to show that QPS-r 
performs as well as maximal matching algorithms empirically no just in theory. 
Hence, we do not compare QPS-r  
with the two recent iterative algorithms in switching~\footnote{Note that they are shown to have reasonably good empirical throughput and delay performance over round-robin-friendly workloads such as uniform and hot-spot traffic when running $1$ or $2$ iterations. However, as described in~\autoref{sec: related-work}, they need to run up to $N$ iterations to provably attain at least 50\% throughput.}: RR/LQF (Round Robin combined with Longest Queue
First)~\cite{Hu2016IterSched} and 
HRF (Highest Rank First)~\cite{Hu_HighestRankFirst_2018}. Instead, 
we compare its 
performance with
that of iSLIP~\cite{McKeown99iSLIP}, a refined and optimized representative parallel maximal matching
algorithm (adapted for switching).  The 
performance of the 
MWM (Maximum Weighted Matching) is also included in the
comparison as a benchmark. Our simulations show conclusively that QPS-1 (running 1 iteration)
performs very well inside 
the provable stability region 
(more precisely, with no more than 50\% offered load), and that QPS-3 (running $3$ iterations) 
has comparable throughput and delay performances
as iSLIP (running $\log_2 N$ iterations), which has a much higher per-port computational 
complexity of $O(\log^2N)$. 

\subsection{Simulation Setup}
\label{subsec: sim-setting}
In our simulations, we 
fix the number of input/output ports, $N$ to $64$.
Later, we investigate how the mean delay performances of these
algorithms scale with respect to $N$ and the findings are reported
in \autoref{subsec: delay-vs-port}.
To measure throughput and delay accurately, we assume each VOQ has an infinite
buffer size and hence there is no packet drop at any input port.
Each simulation run is guided by the following stopping
rule~\cite{Flegal2010MCStopRule,Glynn1992MCStopRule}: The number of time slots simulated is the larger between $500N^2$ and
that is needed
for the difference between the estimated and the actual average delays to be within $0.01$ 
with probability at least $0.98$. 



We assume in our simulations that
each traffic arrival matrix $A(t)$ is Bernoulli {\it i.i.d.} with its traffic
rate matrix $\Lambda$ being equal to the product of the offered load and a traffic pattern
matrix (defined next). Similar Bernoulli arrivals were studied in~\cite{GiacconePrabhakarShah2003SERENA,McKeown99iSLIP,Gong2017QPS}. 
Later, in~\autoref{subsubsec: bursty-arrivals-bench}, we will look at burst traffic arrivals. 
Note that only synthetic traffic (instead of that derived from packet traces) is used in our simulations
because, to the best of our knowledge, there is no meaningful way to combine
packet traces into switch-wide traffic workloads. 
The following four standard types of normalized (with each row or column sum equal to $1$) traffic patterns are used:
(I) \emph{Uniform}: packets arriving at any input port go
to each output port with probability $\frac{1}{N}$.
(II) \emph{Quasi-diagonal}: packets arriving at input port $i$ go to
output port $j \!=\! i$ with probability $\frac{1}{2}$ and go to any other output port
with probability $\frac{1}{2(N-1)}$.
(III) \emph{Log-diagonal}: packets arriving at input port $i$ go
to output port $j = i$ with probability $\frac{2^{(N-1)}}{2^N - 1}$ and
go to any other output port $j$ with probability equal $\frac{1}{2}$ of the
probability of output port $j - 1$ (note: output port $0$ equals output port $N$).
(IV) \emph{Diagonal}: packets arriving at input port $i$ go to
output port $j \!=\! i$ with probability $\frac{2}{3}$, or go to output port
$(i\, \text{mod} \, N) + 1$ with probability $\frac{1}{3}$.
These traffic patterns are listed in order of how skewed the volumes of traffic arrivals to different output
ports are: from uniform being the least skewed, to diagonal being the most skewed.

\subsection{Throughput and Delay Performances}
\label{subsec: benchmark-performance}

\begin{figure}
\centering
\includegraphics[width=\textwidth]{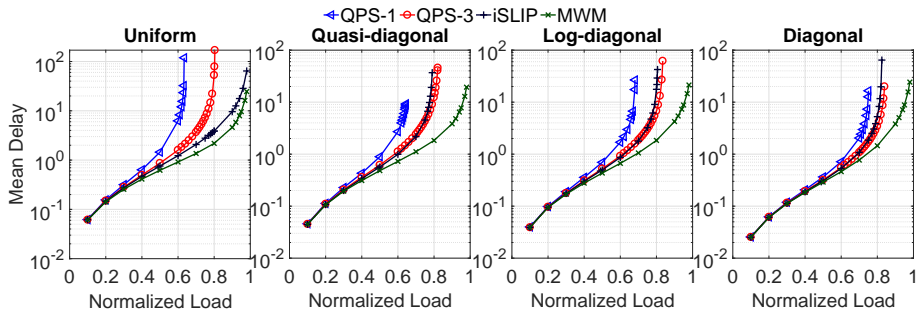}
\caption{Mean delays of QPS-1, QPS-3, iSLIP, and MWM under the $4$ traffic load patterns.}\label{fig: delay-load-bench}
\end{figure}

\subsubsection{Bernoulli Arrivals}\label{subsubsec:bernoulli-arrivals}

We first compare the throughput and delay
performances of QPS-1 (1 iteration), QPS-3 (3 iterations), iSLIP ($\log_2 64$ = 6 iterations), and MWM (length of VOQ as the weight measure). 
Note that we have also investigated how the performance of QPS-r scales with respect to $r$, the results are reported in~\ref{app:diff-r}. 
\autoref{fig: delay-load-bench}
shows their mean delays (in number of time slots)
under the  aforementioned four traffic patterns respectively.
Each subfigure shows how the mean delay (on a {\it log scale} along the y-axis)
varies with the offered load (along the x-axis).
We  make three observations from~\autoref{fig: delay-load-bench}.
First, \autoref{fig: delay-load-bench} clearly shows that, when the offered
load is no larger than $0.5$, QPS-1 has low average delays ({\it i.e.},
more than just being stable) that are close to those of iSLIP and MWM,
under all four traffic patterns.
Second, the maximum sustainable throughputs (where the delays start
to ``go through the roof'' in the subfigures)
of QPS-1 are roughly $0.634, 0.645, 0.681$, and $0.751$ respectively,
under the four traffic patterns respectively;  they are all comfortably
larger than the $50\%$ provable lower bound.
Third, the throughput and delay performances of QPS-3 and iSLIP are comparable:
The former has slightly better delay performances than the latter under all four traffic patterns except the uniform.

\subsubsection{Bursty Arrivals}\label{subsubsec: bursty-arrivals-bench}

In real networks, packet arrivals are likely to be bursty.
In this section, we evaluate the performances of QPS, iSLIP and MWM
under bursty traffic, generated by
a two-state ON-OFF arrival
process, another special case of Markovian arrivals, described in~\cite{Neely2008DelayMM}.
The durations of each ON (burst) stage and OFF (no burst) stage are
geometrically distributed: the probabilities that the ON and OFF states last for $t \ge 0$ time slots are given by
\begin{equation*}
P_{ON}(t) = p(1-p)^t \text{ and } P_{OFF}(t) = q(1-q)^t,
\end{equation*}
with the parameters $p, q \in (0,1)$ respectively. As such, the average duration of the
ON and OFF states are $(1-p)/p$ and $(1-q)/q$ time slots
respectively.

In an OFF state, an incoming packet's destination (\ie output
port) is generated according to the corresponding load matrix. In an
ON state, all incoming packet arrivals to an input port would be destined
to the same output port, thus simulating a burst of packet
arrivals. By controlling $p$, we can control the desired
average burst size while by adjusting $q$,
we can control the load of the traffic.

We have evaluated the mean delay
performances of QPS-3, iSLIP and MWM, with the average burst size
ranging from $16$ to \num{1024} packets. 
We have simulated various offered loads, but here we only present the results, shown in~\autoref{fig: delay-bs-bench}, 
under an offered load of $0.75$; other offered loads lead to similar conclusions. 
\autoref{fig: delay-bs-bench} clearly shows that QPS-3 outperforms iSLIP (under all traffic patterns except the uniform), 
by an increasingly wider margin in both 
absolute and relative terms as the average burst size becomes larger. 

\begin{figure}
\centering
\includegraphics[width=\textwidth]{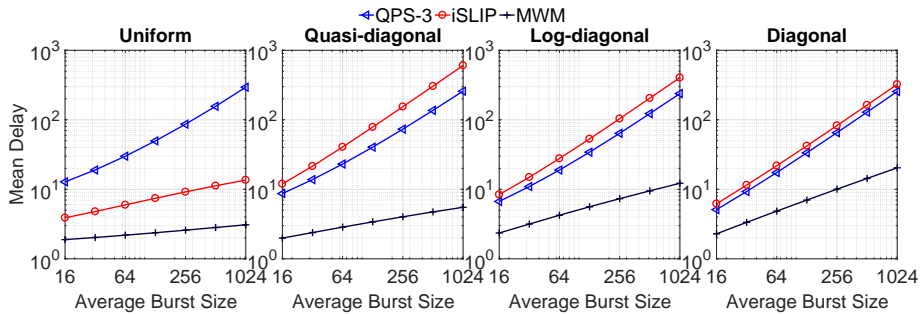}
\caption{Mean delays of QPS-3 against iSLIP and MWM under bursty traffic.}\label{fig: delay-bs-bench}
\end{figure}

\subsection{How Mean Delay Scales with N}
\label{subsec: delay-vs-port}

\autoref{fig: delay-ports-bench} shows how the mean delays of QPS-3, iSLIP (running $\log_2 N$ iterations given any $N$), and MWM scale with the number
of input/output ports $N$, under the four different traffic patterns.
We have simulated the following different  values of $N$: $N\!=\!8, 16, 32, 64, 128, 256, 512$. 
In all these plots, the offered load is $0.75$ (Like in~\autoref{subsubsec: bursty-arrivals-bench}, we have also simulated other 
offered loads. They are omitted here, because they lead to similar conclusions), which is quite high compared to the maximum sustainable throughputs of QPS-3 and iSLIP
(shown in~\autoref{fig: delay-load-bench}) under these four traffic patterns.
\autoref{fig: delay-ports-bench} shows that the mean delays of QPS-3 are slightly lower ({\it i.e.}, better)
than those of iSLIP under all traffic patterns except the uniform.
In addition, the mean delay curves of QPS-3 remain almost flat ({\it i.e.}, constant) under
log-diagonal and diagonal traffic patterns.
Although they increase with $N$ under uniform and quasi-diagonal traffic patterns, they
eventually almost flatten out when $N$ gets larger (say when $N \ge 128$).
These delay curves show that QPS-3, which runs only 3 iterations, deliver slightly better
delay performances, under all traffic patterns except the uniform, than iSLIP (a refined and
optimized parallel maximal matching algorithm adapted for switching), which
runs $\log_2 N$ iterations with each iteration has $O(\log_2 N)$ computational complexity.


\begin{figure}[!ht]
\centering
\includegraphics[width=\textwidth]{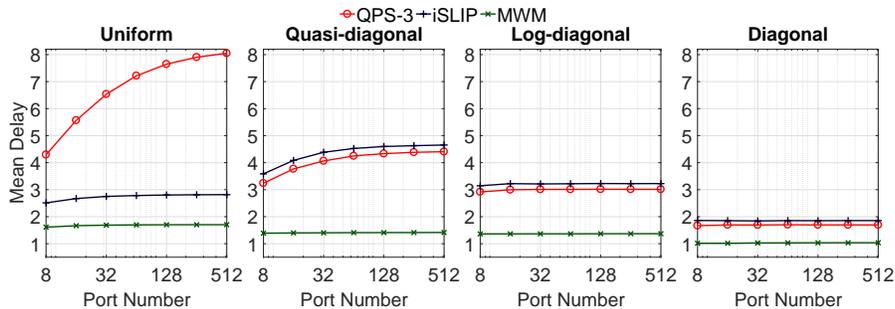}
\caption{Mean delays of QPS-3, iSLIP, and MWM scaling with port number $N$.}\label{fig: delay-ports-bench}
\end{figure}

%% file: sections/related_work.tex
\section{Related Work}\label{sec: related-work}


Scheduling in crossbar switches is a well-studied problem with a large amount of literature. So, in this section, we provide only a brief survey of prior work that is
directly related to ours, 
focusing on those 
we have not described earlier.

\medskip
\noindent
{\bf Iterative algorithms that compute maximal matchings.}
As mentioned earlier, maximal matchings have long been recognized as a cost-effective family in switching. 
Among various types of algorithms that compute maximal matchings, the family of parallel 
iterative algorithms~\cite{Hu_HighestRankFirst_2018,Hu2016IterSched,YihanLi_drr_2001,McKeown_iLQF_1995,Scicchitano_ssr_2007,Lin_SelectiveRequestRound_2011} 
is widely adopted. 
Parallel iterative algorithms 
compute a maximal matching via multiple iterations of message exchanges between the input and 
output ports. Generally, each iteration contains three stages: request, grant, and accept. 
In the request stage, each input port sends requests to output ports. In the grant stage, 
each output port, upon receiving requests from multiple input ports, grants to one. 
Finally, in the accept stage, each input port, upon receiving grants from multiple output ports, 
accepts one. Unfortunately,  all these 
parallel iterative algorithms 
in switching require up to $N$ iterations to guarantee that the resulting matching is a maximal 
matching. In other words, they need up to $N$ iterations to achieve the same provable 
throughput and delay performance guarantees as QPS-1 (running $1$ iteration). 
\medskip
\noindent
{\bf Other algorithms that have performance guarantees.}
Several serial randomized algorithms, starting with TASS~\cite{Tassiulas1998} and culminating
in SERENA~\cite{GiacconePrabhakarShah2003SERENA}, have been proposed
that have a total computational complexity of only $O(N)$ yet can provably attain $100\%$ throughput;
SERENA, the best among them, also delivers a good empirical delay performance.
However, this $O(N)$ complexity is still too high for scheduling high-line-rate high-radix switches,
and none of them has been successfully parallelized ({\it i.e.}, converted to a parallel iterative algorithm)
yet, except that SERENA was recently parallelized in~\cite{Gong_SERENADEParallelIterative_2020} that 
reduces the $O(N)$ computational complexity to $O(\log N)$ per port, 
which is still quite high for high-line-rate high-radix switches. 

In~\cite{Ye10Variable},
a crossbar scheduling algorithm specialized for switching variable-length packets was
proposed, that has $O(1)$ total computational complexity. 
Although this algorithm can provably attain $100\%$ throughput, its delay performance is poor.
For example, as shown in~\cite{Gong2017QPS},
its average delays, under the aforementioned four
standard traffic matrices, are roughly
$3$ orders of magnitudes higher than those of
SERENA~\cite{GiacconePrabhakarShah2003SERENA} even under a moderate offered load of $0.6$.

%% file: sections/conclusion.tex
\section{Conclusion}\label{sec:conclusion}


%
In this work, we propose QPS-r, a parallel iterative switching algorithm
with $O(1)$ computational complexity per port. We prove, through Lyapunov stability analysis,
that it achieves the same throughput and delay guarantees in theory,
and demonstrate through simulations that it has comparable performances in practice
as the family of maximal matching algorithms (adapted for switching);  maximal matching algorithms
are much more expensive computationally (at least $O(\log N)$ iterations
and a total of $O(\log^2 N)$ per-port computational
complexity). These salient properties make QPS-r an excellent candidate
algorithm that is fast enough computationally and can
deliver acceptable throughput and delay performances for high-link-rate high-radix switches.

%% file: sections/appendix.tex
\section{Irreducibility of \texorpdfstring{$\{Q(t)\}_{t=0}^{\infty}$}{Q(t)}}\label{app:proof-markov-for-qps-1}


We have shown that the queueing process $\{Q(t)\}_{t=0}^{\infty}$ is a Markov chain in the proof of~\autoref{coro:stability-qps}.
Now we prove that this Markov is irreducible using the same approach as used in~\cite{Gong2017QPS}. 
Here we provide only a proof sketch. 
We show that the queueing process $Q(t)$, starting from any 
state ({\it i.e.,} VOQ lengths) it is currently in, will with a positive probability return to the 
$\mathbf{0}$ state, in which all VOQs have a length of $0$, in a finite number of steps ({\it i.e.,} time slots). 
To show this property, we claim that, for any integer $\tau > 0$, the switch could, with a nonzero probability,
have no packet arrivals to any of its VOQs during time slots $t, t+1, \cdots, t + \tau$.  
This claim is true because, the arrival process $A(t)$ is {\it i.i.d.}, and for any $1\le i\le N$ and $1\le j\le N$, we have 
$\beta_{ij} \triangleq P[a_{ij}(t) = 0] > 0$ (since the maximum load factor $\rho <1/2$).  
Hence, when there are no packet arrivals during time slots $t, t+1, \cdots, t + \tau$, which happens with 
a nonzero probability, QPS-1 can clear all the queues during this period, with a sufficiently large $\tau$, and return the queueing process $Q(t)$ to the $\mathbf{0}$ state. 
Therefore, the Markov chain is irreducible.  

\section{Proof of \autoref{thm:stability-tool-map}} \label{app:proof-stability}

In this section, we prove \autoref{thm:stability-tool-map}. The proof follows the same outline as the proof of
\autoref{thm:stability-tool} with some necessary modifications to allow for Markovian arrivals. 
Note that we prove only (I) here, since \autoref{thm:delay-tool-map} that we will shortly prove
implies (II).  

For notational convenience, we define $Z(t)\!\triangleq\!\big{(}Q(t),X(t)\big{)}$. 
Unlike in the proof of~\autoref{thm:stability-tool} where we use the 1-step Foster-Lyapunov stability criterion~\cite[Proposition 2.1.1]{brucelecture2006}, 
here we use the $n$-step (multiple-step) Foster-Lyapunov stability criterion~\cite[Theorem 2.2.4]{fayolle1995topics}. 
We consider again the following Lyapunov function: $L\big{(}Z(t)\big{)}\!=\!\sum_{i,j}q_{ij}(t)Q^\dagger_{ij}(t)$, and 
show that its $n$-step drift, which is defined as $\mathbf{E}\big{[}L\big{(}Z(t+n)\big{)} - L\big{(}Z(t)\big{)} \mid Z(t)\big{]}$, is negative except in a finite set.  

To compute the $n$-step drift, we first compute,
\begin{equation}\label{eq:n-step-drfit-inside}
L\big{(}Z(t+n)\big{)} - L\big{(}Z(t)\big{)}=\sum_{i,j}q_{ij}(t+n)Q^\dagger_{ij}(t+n) - \sum_{i,j}q_{ij}(t)Q^\dagger_{ij}(t). 
\end{equation}

Using~\labelcref{eq:queueing-dynamics-entrywise,eq:dagger-def}, and definitions of $D^\dagger_{ij}(t)$ and 
$A^\dagger_{ij}(t)$, we have,
\begin{equation}\label{eq:n-step-queue-evo}
q_{ij}(t+n) = q_{ij}(t) - \sum_{k=0}^{n-1}d_{ij}(t+k) + \sum_{k=0}^{n-1}a_{ij}(t+k),
\end{equation}
and
\begin{equation}\label{eq:n-step-queue-evo-dagger}
Q^\dagger_{ij}(t+n) = Q^\dagger_{ij}(t) - \sum_{k=0}^{n-1}D^\dagger_{ij}(t+k) + \sum_{k=0}^{n-1}A^\dagger_{ij}(t+k).
\end{equation}

Thus, the term $q_{ij}(t+n)Q^\dagger_{ij}(t+n)$ in the first term on the RHS of~\autoref{eq:n-step-drfit-inside} can be written as
\begin{align}
&q_{ij}(t+n)Q^\dagger_{ij}(t+n) \notag\\
=& q_{ij}(t)Q^\dagger_{ij}(t)+q_{ij}(t)\big{(}\sum_{k=0}^{n-1}A^\dagger_{ij}(t+k)-\sum_{k=0}^{n-1}D^\dagger_{ij}(t+k)\big{)} \notag \\
&{}+ Q^\dagger_{ij}(t)\big{(}\sum_{k=0}^{n-1}a_{ij}(t+k)-\sum_{k=0}^{n-1}d_{ij}(t+k)\big{)} \notag\\
&{}+ 
\big{(}\sum_{k=0}^{n-1}a_{ij}(t+k)-\sum_{k=0}^{n-1}d_{ij}(t+k)\big{)} \big{(}\sum_{k=0}^{n-1}A^\dagger_{ij}(t+k)-\sum_{k=0}^{n-1}D^\dagger_{ij}(t+k)\big{)}.  \label{eq:qnprod-to-qprod}
\end{align}

Substituting \autoref{eq:qnprod-to-qprod} into ~\autoref{eq:n-step-drfit-inside}, we have
\begin{align}
&L\big{(}Z(t+n)\big{)} - L\big{(}Z(t)\big{)} \notag \\
=& \sum_{i,j}q_{ij}(t)\big{(}\sum_{k=0}^{n-1}A^\dagger_{ij}(t+k)-\sum_{k=0}^{n-1}D^\dagger_{ij}(t+k)\big{)} \notag \\
&{}+ \sum_{i,j}Q^\dagger_{ij}(t)\big{(}\sum_{k=0}^{n-1}a_{ij}(t+k)-\sum_{k=0}^{n-1}d_{ij}(t+k)\big{)} \notag\\
&{}+ 
\sum_{i,j}\big{(}\sum_{k=0}^{n-1}a_{ij}(t+k)-\sum_{k=0}^{n-1}d_{ij}(t+k)\big{)} \big{(}\sum_{k=0}^{n-1}A^\dagger_{ij}(t+k)-\sum_{k=0}^{n-1}D^\dagger_{ij}(t+k)\big{)}. \label{eq:n-step-drfit-inside-2}
\end{align}

Focusing on the second term above and using \autoref{eq:dagger-def}, we have
\begingroup
\allowdisplaybreaks
\begin{align}
&\sum_{i,j}Q^\dagger_{ij}(t)\big{(}\sum_{k=0}^{n-1}a_{ij}(t+k)-\sum_{k=0}^{n-1}d_{ij}(t+k)\big{)} \notag \\
=&\sum_{i,j}\sum_{(l,w)\in Q^\dagger_{ij}}q_{lw}(t)\big{(}\sum_{k=0}^{n-1}a_{ij}(t+k)-\sum_{k=0}^{n-1}d_{ij}(t+k)\big{)}  \notag\\
=&\sum_{l,w}\sum_{(i,j)\in Q^\dagger_{lw}}q_{lw}(t)\big{(}\sum_{k=0}^{n-1}a_{ij}(t+k)-\sum_{k=0}^{n-1}d_{ij}(t+k)\big{)}  \notag\\
=&\sum_{l,w}q_{lw}\big{(}\sum_{(i,j)\in Q^\dagger_{lw}}\sum_{k=0}^{n-1}a_{ij}(t+k)-\sum_{(i,j)\in Q^\dagger_{lw}}\sum_{k=0}^{n-1}d_{ij}(t+k)\big{)} \notag \\
=&\sum_{l,w}q_{lw}\big{(}\sum_{k=0}^{n-1}\sum_{(i,j)\in Q^\dagger_{lw}}a_{ij}(t+k)-\sum_{k=0}^{n-1}\sum_{(i,j)\in Q^\dagger_{lw}}d_{ij}(t+k)\big{)} \notag \\
=&\sum_{l,w}q_{lw}\big{(}\sum_{k=0}^{n-1}A^\dagger_{lw}(t+k)-\sum_{k=0}^{n-1}D^\dagger_{lw}(t+k)\big{)} \notag \\
=&\sum_{i,j}q_{ij}(t)\big{(}\sum_{k=0}^{n-1}A^\dagger_{ij}(t+k)-\sum_{k=0}^{n-1}D^\dagger_{ij}(t+k)\big{)}.
\end{align}
\endgroup

Hence, the drift can be written as
\begin{align}
&\mathbf{E}\big{[}L\big{(}Z(t+n)\big{)} - L\big{(}Z(t)\big{)} \mid Z(t)\big{]}\nonumber\\
=&\mathbf{E}\Big{[}2\sum_{i,j}q_{ij}(t)\big{(}\sum_{k=0}^{n-1}A^\dagger_{ij}(t+k)-\sum_{k=0}^{n-1}D^\dagger_{ij}(t+k)\big{)} \mid Z(t)\Big{]} \notag \\
&{}\!+\!\mathbf{E}\Big{[}\sum_{i,j}\big{(}\sum_{k=0}^{n-1}a_{ij}(t\!+\!k)\!-\!\sum_{k=0}^{n-1}d_{ij}(t\!+\!k)\big{)} \big{(}\sum_{k=0}^{n-1}A^\dagger_{ij}(t\!+\!k)\!-\!\sum_{k=0}^{n-1}D^\dagger_{ij}(t\!+\!k)\big{)}\!\mid\! Z(t)\Big{]}. \label{eq:raw-drift-map}
\end{align}


Since $0\le a_{ij}(t)\le a_{max}$ and $0\le d_{ij}(t)\le 1$, for any $i,j$ and $t$, we have 
\begin{align}
&\mathbf{E}\Big{[}\sum_{i,j}\big{(}\sum_{k=0}^{n-1}a_{ij}(t\!+\!k)\!-\!\sum_{k=0}^{n-1}d_{ij}(t\!+\!k)\big{)} \big{(}\sum_{k=0}^{n-1}A^\dagger_{ij}(t\!+\!k)\!-\!\sum_{k=0}^{n-1}D^\dagger_{ij}(t\!+\!k)\big{)}\!\mid\! Z(t)\Big{]} \notag \\
\le&\mathbf{E}\Big{[}\sum_{i,j}\big{(}\sum_{k=0}^{n-1}a_{ij}(t\!+\!k)\!+\!\sum_{k=0}^{n-1}d_{ij}(t\!+\!k)\big{)} \big{(}\sum_{k=0}^{n-1}A^\dagger_{ij}(t\!+\!k)\!+\!\sum_{k=0}^{n-1}D^\dagger_{ij}(t\!+\!k)\big{)}\!\mid\! Z(t)\Big{]} \notag \\
\le&\mathbf{E}\Big{[}\sum_{i,j}\big{(}na_{max}+n\big{)} \big{(}(2N-1)na_{max}+(2N-1)n\big{)}\!\mid\! Z(t)\Big{]} \notag \\
 =& C_1(n). \label{eq:term-2-tbp-map}
\end{align}
where $C_1(n)\triangleq (2N-1)n^2N^2(a_{max}+1)^2$.

Now we claim there exists some $n>K_\rho$ (where $K_\rho$ is defined in~\autoref{fact:maximum-load-bounded}) such that the following inequality holds, which we will prove shortly.
\begin{align}
&\mathbf{E}\Big{[}2\sum_{i,j}q_{ij}(t)\big{(}\sum_{k=0}^{n-1}A^\dagger_{ij}(t+k)-\sum_{k=0}^{n-1}D^\dagger_{ij}(t+k)\big{)} \mid Z(t)\Big{]}\notag \\
\le& 2n\Big{(}\xi_\rho - 1 + \frac{(2N-1)K_\rho a_{max}-K_\rho\xi_\rho}{n}\Big{)}\|Q(t)\|_1 + C_2(n), \label{eq:term-1-tbp-map}
\end{align}
where $K_\rho$ is defined in~\autoref{fact:maximum-load-bounded},  
$\xi_\rho\triangleq \max\{\xi(k): k\in\{K_\rho,K_\rho+1,\cdots,n-1\}\}$ where $\xi(\cdot)$ is defined~\autoref{eq:load-bound} and $C_2(n)\triangleq (2a_{max} +1)(n-1)n N^2$.

Define $n_0\triangleq \min\{n\in \mathbb{N}_+\mbox{ and } n > K_\rho: \xi_\rho - 1 + \frac{(2N-1)K_\rho a_{max}-K_\rho\xi_\rho}{n} < 0\}$. Such 
an $n_0$ should exist because $\xi_\rho<1$ and $\frac{(2N-1)K_\rho a_{max}-K_\rho\xi_\rho}{n}\rightarrow 0$ as $n\rightarrow \infty$. 
Define $\varepsilon(n_0) \triangleq -2n_0\big{(}\xi_\rho - 1 + \frac{(2N-1)K_\rho a_{max}-K_\rho\xi_\rho}{n_0}\big{)}$, clearly $\varepsilon(n_0)>0$. 
With~\autoref{eq:term-2-tbp-map} and~\autoref{eq:term-1-tbp-map} substituted into~\autoref{eq:raw-drift-map}, we have

\begin{align}
&\mathbf{E}\big{[}L\big{(}Z(t+n_0)\big{)} - L\big{(}Z(t)\big{)} \mid Z(t)\big{]} \notag\\
\leq& -\varepsilon(n_0) \|Q(t)\|_1 + C_1(n_0) + C_2(n_0) \notag \\
\le& -\varepsilon(n_0) (\|Z(t)\|_1 - \chi N^2) + C_1(n_0) + C_2(n_0), \label{eq:x-t-1-norm} 
\end{align}
where $\chi\triangleq\max_{i,j}\{\chi_{ij}\}$. 

Inequality~\autoref{eq:x-t-1-norm} is 
due to $\varepsilon(n_0) > 0$ and $x_{ij}(t)\le \chi_{ij}$ (since each $\{x_{ij}(t)\}_{t=0}^\infty$ is a DTMC on a finite-state space of $\{1,2,\cdots,\chi_{ij}\}$). Because $\varepsilon(n_0)>0$, there exist $B,\epsilon >0$ such that, whenever $\|Z(t)\|_1 \!>\! B$,
\begin{equation*}
\mathbf{E}\big{[}L\big{(}Z(t+n_0)\big{)} - L\big{(}Z(t)\big{)} \mid Z(t)\big{]} \leq -\epsilon.
\end{equation*}

So the theorem (\autoref{thm:stability-tool-map}) follows using $n_0$-step Foster-Lyapunov criterion~\cite[Theorem 2.2.4]{fayolle1995topics}.

Now we proceed to prove~\autoref{eq:term-1-tbp-map}. By simple calculations, we have
\begin{align}
&\mathbf{E}\Big{[}2\sum_{i,j}q_{ij}(t)\big{(}\sum_{k=0}^{n-1}A^\dagger_{ij}(t+k)-\sum_{k=0}^{n-1}D^\dagger_{ij}(t+k)\big{)} \mid Z(t)\Big{]}\notag \\
=&2\Big{(}\sum_{i,j}\mathbf{E}\Big{[}q_{ij}(t)\sum_{k=0}^{n-1}A^\dagger_{ij}(t+k)\mid Z(t) \Big{]}\!-\!\sum_{i,j}\mathbf{E}\Big{[}q_{ij}(t)\sum_{k=0}^{n-1}D^\dagger_{ij}(t+k)\mid Z(t)\Big{]}\Big{)}. \label{eq:bounded-first-term-map}
\end{align}

Focusing on the first term in the parentheses above, for the conditional expectation $\mathbf{E}\Big{[}q_{ij}(t)\sum_{k=0}^{n-1}A^\dagger_{ij}(t+k)\mid Z(t) \Big{]}$, we have

\begin{align}
&\mathbf{E}\Big{[}q_{ij}(t)\sum_{k=0}^{n-1}A^\dagger_{ij}(t+k)\mid Z(t) \Big{]} \notag \\
=&\mathbf{E}\Big{[}q_{ij}(t)\big{(}\sum_{k=0}^{K_\rho-1}A^\dagger_{ij}(t+k) + \sum_{k=K_\rho}^{n-1}A^\dagger_{ij}(t+k)\big{)}\mid Z(t) \Big{]} \notag \\
=&\mathbf{E}\Big{[}q_{ij}(t)\sum_{k=0}^{K_\rho-1}A^\dagger_{ij}(t+k) \mid Z(t)\Big{]}  + \mathbf{E}\Big{[}q_{ij}(t)\sum_{k=K_\rho}^{n-1}A^\dagger_{ij}(t+k)\big{)}\mid Z(t) \Big{]}. \label{eq:temp-1} 
\end{align}

Since $0\le a_{ij}(t)\le a_{max}$ for any $i,j$ and $t$, we have the first expectation above, $\mathbf{E}\Big{[}q_{ij}(t)\sum_{k=0}^{K_\rho-1}A^\dagger_{ij}(t+k) \mid Z(t)\Big{]}$, is bounded by 
\begin{equation}\label{eq:temp-1-bound-1}
\mathbf{E}\Big{[}q_{ij}(t)\sum_{k=0}^{K_\rho-1}A^\dagger_{ij}(t+k) \mid Z(t)\Big{]} \le (2N-1)K_\rho a_{max}q_{ij}(t).
\end{equation}

Using~\autoref{eq:load-bound},~\autoref{fact:maximum-load-bounded},~and~\autoref{fact:dependent-arrival-rate-map}, 
we have, for any $k\ge K_\rho$,
\begin{align*}
&\mathbf{E}\Big{[}q_{ij}(t)A^\dagger_{ij}(t+k) \mid Z(t) \Big{]} \notag \\
=&q_{ij}(t)\mathbf{E}\Big{[}A^\dagger_{ij}(t+k) \mid Z(t) \Big{]} \notag \\
\le&q_{ij}(t)\big{(}\Lambda_{ij}^\dagger + \sum_{(l,w)\in Q^\dagger_{ij}}\zeta_{lw}(k)\big{)} \notag \\
\le&\xi(k)q_{ij}(t) \notag \\
\le& \xi_\rho q_{ij}(t).
\end{align*}

Therefore, the second expectation on the RHS of~\autoref{eq:temp-1}, $\mathbf{E}\Big{[}q_{ij}(t)\sum_{k=K_\rho}^{n-1}A^\dagger_{ij}(t+k)\big{)}\mid Z(t) \Big{]}$, can be bounded by
\begin{equation}\label{eq:temp-1-bound-2}
\mathbf{E}\Big{[}q_{ij}(t)\sum_{k=K_\rho}^{n-1}A^\dagger_{ij}(t+k)\big{)}\mid Z(t) \Big{]} \le (n-K_\rho)\xi_\rho q_{ij}(t).
\end{equation}

Now, we proceed to bound the second term in the parentheses on the RHS of~\autoref{eq:bounded-first-term-map}, 
$\sum_{i,j}\mathbf{E}\Big{[}q_{ij}(t)\sum_{k=0}^{n-1}D^\dagger_{ij}(t+k)\mid Z(t)\Big{]}$. 

Using~\autoref{eq:queueing-dynamics-entrywise}, we have $q_{ij}(t+k)=q_{ij}(t)-\sum_{\tau=0}^{k-1}d_{ij}(t+\tau)+\sum_{\tau=0}^{k-1}a_{ij}(t+\tau)$ for any integer $k\ge 0$. 
Because $d_{ij}(t)\ge 0$ and $a_{ij}(t)\le a_{max}$ for any $i,j$ and $t$, 
$q_{ij}(t)=q_{ij}(t+k)+\sum_{\tau=0}^{k-1}d_{ij}(t+\tau)-\sum_{\tau=0}^{k-1}a_{ij}(t+\tau)\ge q_{ij}(t+k)-ka_{max}$ for any integer $k\ge 0$. Thus, for any integer $k\ge 0$, we have 
\begingroup
\allowdisplaybreaks
\begin{align}
&\sum_{i,j}\mathbf{E}\Big{[}q_{ij}(t)D^\dagger_{ij}(t+k)\mid Z(t)\Big{]} \notag \\
\ge&\sum_{i,j}\mathbf{E}\Big{[}\big{(}q_{ij}(t+k)-ka_{max}\big{)}D^\dagger_{ij}(t+k)\mid Z(t)\Big{]} \notag \\
=&\sum_{i,j}\Bigg{(}\mathbf{E}\Big{[}q_{ij}(t+k)D^\dagger_{ij}(t+k)\mid Z(t)\Big{]} - \mathbf{E}\Big{[}k a_{max}D^\dagger_{ij}(t+k)\mid Z(t)\Big{]}\Bigg{)}  \notag \\
\ge&\sum_{i,j}\Bigg{(}\mathbf{E}\Big{[}q_{ij}(t+k)D^\dagger_{ij}(t+k)\mid Z(t)\Big{]} -2ka_{max}\Bigg{)}  \label{eq:qd-bound-1} \\
=&\sum_{i,j}\mathbf{E}\Big{[}q_{ij}(t+k)D^\dagger_{ij}(t+k)\mid Z(t)\Big{]} - 2ka_{max}N^2 \notag \\
=&\sum_{i,j}\mathbf{E}\Bigg{[}\mathbf{E}\Big{[}\big{(}q_{ij}(t+k)D^\dagger_{ij}(t+k) \mid Q(t+k)\Big{]}\mid Z(t)\Bigg{]} - 2ka_{max}N^2 \notag \\
=&\mathbf{E}\Bigg{[}\sum_{i,j}\mathbf{E}\Big{[}q_{ij}(t+k)D^\dagger_{ij}(t+k) \mid Q(t+k)\Big{]}\mid Z(t)\Bigg{]} - 2ka_{max}N^2 \notag \\
\ge&\mathbf{E}\Big{[}\sum_{i,j}q_{ij}(t+k)\mid Z(t) \Big{]}- 2ka_{max}N^2 \label{eq:qd-bound-2} \\
=&\mathbf{E}\Big{[}\sum_{i,j}\big{(}q_{ij}(t)-\sum_{\tau=0}^{k-1}d_{ij}(t+\tau)+\sum_{\tau=0}^{k-1}a_{ij}(t+\tau)\big{)}\mid Z(t) \Big{]}- 2ka_{max}N^2 \notag \\
\ge& \mathbf{E}\Big{[}\sum_{i,j}\big{(}q_{ij}(t) -k\big{)}\mid Z(t) \Big{]} - 2ka_{max}N^2 \label{eq:qd-bound-3} \\
=&\sum_{i,j}\big{(}q_{ij}(t) -k\big{)} - 2ka_{max}N^2 \notag \\
=&\|Q(t)\|_1 - (2a_{max} +1)kN^2. \label{eq:qd-bound-before-final}
\end{align}
\endgroup 

Inequality~\autoref{eq:qd-bound-1} is due to $D^\dagger_{ij}(t)\le 2$ for any $i,j$ and $t$. 
Inequality~\autoref{eq:qd-bound-2} is due to \autoref{eq:QPS-property} and 
inequality~\autoref{eq:qd-bound-3} is due to $d_{ij}(t)\le 1$ and $a_{ij}(t)\ge 0$ for any $i,j$ and $t$.

Using~\autoref{eq:qd-bound-before-final}, we have
\begingroup
\allowdisplaybreaks
\begin{align}
&\sum_{i,j}\mathbf{E}\Big{[}q_{ij}(t)\sum_{k=0}^{n-1}D^\dagger_{ij}(t+k)\mid Z(t)\Big{]} \notag \\
=&\sum_{k=0}^{n-1}\sum_{i,j}\mathbf{E}\Big{[}q_{ij}(t)D^\dagger_{ij}(t+k)\mid Z(t)\Big{]} \notag \\
\ge&\sum_{k=0}^{n-1}\big{(} \|Q(t)\|_1 - (2a_{max} +1)kN^2\big{)} \notag \\
=&n \|Q(t)\|_1 - (2a_{max} +1)(n-1)n N^2 / 2 \label{eq:qd-bound-final}
\end{align}
\endgroup 

Now, using \labelcref{eq:bounded-first-term-map,eq:temp-1} and inequalities~\labelcref{eq:temp-1-bound-1,eq:temp-1-bound-2,eq:qd-bound-final}, we have
\begin{align}
&\mathbf{E}\Big{[}2\sum_{i,j}q_{ij}(t)\big{(}\sum_{k=0}^{n-1}A^\dagger_{ij}(t+k)-\sum_{k=0}^{n-1}D^\dagger_{ij}(t+k)\big{)} \mid Z(t)\Big{]}\notag \\
\le&2\Bigg{(}(2N-1)K_\rho a_{max}\|Q(t)\|_1 + (n-K_\rho)\xi_\rho\|Q(t)\|_1 \notag \\
&{}- \Big{(}n \|Q(t)\|_1 - (2a_{max} +1)(n-1)n N^2 / 2\Big{)}\Bigg{)} \notag \\
=&2n\Big{(}\xi_\rho - 1 + \frac{(2N-1)K_\rho a_{max}-K_\rho\xi_\rho}{n}\Big{)}\|Q(t)\|_1 + C_2(n),
\end{align}
where $C_2(n)=(2a_{max} +1)(n-1)n N^2$.

\section{\autoref{thm:stability-tool-map} Applies to QPS-1}\label{app:proof-markov-map-for-qps-1}

In this section, we prove that \autoref{thm:stability-tool-map} applies to QPS-1. More precisely, 
under Markovian arrivals whenever the maximum load factor $\rho < 1/2$, QPS-1 is stable in the following sense: 
The resulting joint process $\{(Q(t),X(t)\}_{t=0}^{\infty}$ 
is a positive recurrent Markov chain and the stationary distribution $\bar{Q}$ of the 
queueing process $\{(Q(t)\}_{t=0}^{\infty}$ has finite first moment. 

The joint process $\{(Q(t),X(t)\}_{t=0}^{\infty}$ is a Markov chain, by the following two facts.
First, $X(t)$ is a Markov chain so $X(t)$ only depends on $X(t-1)$. Second,
by~\autoref{eq:queueing-dynamics-entrywise}, $Q(t)$ depends on only $Q(t-1), A(t-1)$ and
$D(t-1)$, where $A(t-1)$ is a function of only $X(t-1)$ and $D(t-1)$ is a function of only
$Q(t-1)$.  

The reasoning for the irreducibility of this Markov chain $\{(Q(t),X(t)\}_{t=0}^{\infty}$ 
is almost identical to those in~\ref{app:proof-markov-for-qps-1}, so we omit it here for brevity.


\section{Proof of \autoref{thm:delay-tool-map}}\label{app:proof-delay}

\noindent
{\bf Part (I).} 
Now we prove (I) in \autoref{thm:delay-tool-map}. 
Like in~\ref{app:proof-stability}, we let $Z(t)\triangleq\left(Q(t),X(t)\right)$. 
Like in the proof of~\autoref{thm:delay-tool}, we replace function $V$ in~\autoref{thm:moment-bound} by $L$, the Lyapunov function used in the proof 
of~\autoref{thm:stability-tool-map} and $Y_t$ by $Z(t)$.
We define $f$ and $g$ functions 
in the same way as in the proof of~\autoref{thm:delay-tool}. 
First, we show that the $Y_t$, $V$, $f$, and $g$ defined above satisfy~\autoref{eq:mbreq}.

As derived in~\autoref{subsec:throughput-analysis}, the LHS of~\autoref{eq:mbreq} can be written as
\begingroup
\allowdisplaybreaks
\begin{align*}
&V\big{(}Y_{t+1}\big{)} - V\big{(}Y_t\big{)} \nonumber\\
=&L\big{(}Z(t+1)\big{)} - L\big{(}Z(t)\big{)} \\
=&2\sum_{i,j}q_{ij}(t)\big{(}A^\dagger_{ij}(t)-D^\dagger_{ij}(t)\big{)} 
+\sum_{i,j}\big{(}a_{ij}(t)\!-\!d_{ij}(t)\big{)}\big{(}A^\dagger_{ij}(t)\!-\!D^\dagger_{ij}(t)\big{)} \\
=& -\Big{(}-2\sum_{i,j}q_{ij}(t)\big{(}A^\dagger_{ij}(t)-D^\dagger_{ij}(t)\big{)} + h(Y_t)\Big{)} \\
&{}+\Big{(}\sum_{i,j}\big{(}a_{ij}(t)\!-\!d_{ij}(t)\big{)}\big{(}A^\dagger_{ij}(t)\!-\!D^\dagger_{ij}(t)\big{)} + h(Y_t) \Big{)} \\
=&-f(Y_t) + g(Y_t).
\end{align*}
\endgroup
Therefore, \autoref{eq:mbreq} is satisfied. 

Now we claim the following two inequalities that hold in steady state of the positive recurrent Markov chain $\{Z(t)\}_{t=0}^\infty$, which we will prove shortly.
\begin{align}
&\mathbf{E}\big{[}2{\textstyle \sum\limits_{i,j}}q_{ij}(t)\big{(}A^\dagger_{ij}(t)\!-\!D^\dagger_{ij}(t)\big{)}\big{]} \notag \\
\le& 2\Bigg{(}(\xi(K_\rho) \!-\! 1)\mathbf{E}[\|\bar{Q}\|_1] \!+\!{\textstyle\sum\limits_{i,j}}\mathbf{E}\Big{[}{\textstyle \sum\limits_{k=0}^{K_\rho-1}}a_{ij}(k-K_\rho)A^\dagger_{ij}(0)\Big{]}\Bigg{)}, \label{eq:term-1-tbp-map-delay}
\end{align}
\begin{equation}\label{eq:term-2-tbp-map-delay}
\mathbf{E}\big{[}{\textstyle \sum\limits_{i,j}}\big{(}a_{ij}(t)\!-\!d_{ij}(t)\big{)}\big{(}A^\dagger_{ij}(t)\!-\!D^\dagger_{ij}(t)\big{)}\big{]} \leq \sum_{i,j}\left(\mathbf{E}[a_{ij}(0)A^\dagger_{ij}(0)] +\lambda_{ij}\right),
\end{equation}
where $\xi(\cdot)$ and $K_\rho$ are defined in~\autoref{eq:load-bound} and~\autoref{fact:maximum-load-bounded} respectively.

Then, using similar reasoning as in the proof of~\autoref{thm:delay-tool}, we have
\begingroup
\allowdisplaybreaks
\begin{align*}
&-2\left((\xi(K_\rho)-1)\mathbf{E}[\|\bar{Q}\|_1]+\sum_{i,j}\mathbf{E}\left[\sum_{k=0}^{K_\rho-1}a_{ij}(k-K_\rho)A^\dagger_{ij}(0)\right]\right)\nonumber\\
\le&\mathbf{E}[f(\bar{Y}) - h(\bar{Y})] \\
\le&\mathbf{E}[g(\bar{Y}) - h(\bar{Y})] \\
\le&\sum_{i,j}\left(\mathbf{E}[a_{ij}(0)A^\dagger_{ij}(0)] +\lambda_{ij}\right). 
\end{align*}
\endgroup

Therefore, we have, in steady state,
\begin{align*}
\mathbf{E}[\|\bar{Q}\|_1]&\leq\frac{1}{2(1-\xi(K_\rho))}\left(\sum_{i,j}\left(\mathbf{E}[a_{ij}(0)A^\dagger_{ij}(0)] +\lambda_{ij}\right) \right.\notag \\
&\qquad{}\left.+ 2\sum_{i,j}\mathbf{E}\left[\sum_{k=0}^{K_\rho-1}a_{ij}(k-K_\rho)A^\dagger_{ij}(0)\right]\right). 
\end{align*}

Now we proceed to prove~\autoref{eq:term-1-tbp-map-delay}.
\begin{align}
&\mathbf{E}\big{[}2{\textstyle \sum\limits_{i,j}}q_{ij}(t)\big{(}A^\dagger_{ij}(t)\!-\!D^\dagger_{ij}(t)\big{)}\big{]} \notag \\
=&2\Big{(}\mathbf{E}\big{[}\sum_{i,j}q_{ij}(t)A^\dagger_{ij}(t)\big{]} - \mathbf{E}\big{[}\sum_{i,j}q_{ij}(t)D^\dagger_{ij}(t)\big{]}\Big{)}. \label{eq:map-delay-bound-1}
\end{align}

Focusing on the first expectation term, $\mathbf{E}\big{[}\sum_{i,j}q_{ij}(t)A^\dagger_{ij}(t)\big{]}$, above, we have
\begingroup
\allowdisplaybreaks
\begin{align}
&\mathbf{E}\Big{[}\sum_{i,j}q_{ij}(t)A^\dagger_{ij}(t)\Big{]} \notag \\
=&\mathbf{E}\Big{[}\sum_{i,j}\big{(}q_{ij}(t\!-\!K_\rho)\!-\!\sum_{k=0}^{K_\rho-1}d_{ij}(t\!-\!K_\rho\!+\!k)\!+\!\sum_{k=0}^{K_\rho-1}a_{ij}(t\!-\!K_\rho+k)\big{)}A^\dagger_{ij}(t)\Big{]} \label{eq:map-delay-bound-1-1} \\
\le&\mathbf{E}\Big{[}\sum_{i,j}\big{(}q_{ij}(t-K_\rho)+\sum_{k=0}^{K_\rho-1}a_{ij}(t-K_\rho+k)\big{)}A^\dagger_{ij}(t)\Big{]} \label{eq:map-delay-bound-1-2} \\
=&\mathbf{E}\Big{[}\sum_{i,j} q_{ij}(t-K_\rho)A^\dagger_{ij}(t) \Big{]} + \mathbf{E}\Big{[}\sum_{i,j}\sum_{k=0}^{K_\rho-1}a_{ij}(t-K_\rho+k)A^\dagger_{ij}(t)\Big{]} \notag \\
=&\mathbf{E}\Big{[}\sum_{i,j} q_{ij}(t-K_\rho)A^\dagger_{ij}(t) \Big{]} + \sum_{i,j}\mathbf{E}\Big{[}\sum_{k=0}^{K_\rho-1}a_{ij}(t-K_\rho+k)A^\dagger_{ij}(t)\Big{]} \notag \\
=&\mathbf{E}\Big{[}\sum_{i,j} q_{ij}(t\!-\!K_\rho)\mathbf{E}\big{[}A^\dagger_{ij}(t) \mid H(t\!-\!K_\rho) \big{]}\Big{]} \!+\! \sum_{i,j}\mathbf{E}\Big{[}\sum_{k=0}^{K_\rho-1}a_{ij}(t\!-\!K_\rho+k)A^\dagger_{ij}(t)\Big{]} \notag \\
\le&\mathbf{E}\Big{[}\sum_{i,j} q_{ij}(t\!-\!K_\rho)\big{(}\Lambda^\dagger_{ij}(t) \!+\! \sum_{(l,w)\in Q_{ij}^\dagger}\zeta_{lw}(K_\rho) \big{)}\Big{]} \notag \\
&\qquad{}+ \sum_{i,j}\mathbf{E}\Big{[}\sum_{k=0}^{K_\rho-1}a_{ij}(t\!-\!K_\rho+k)A^\dagger_{ij}(t)\Big{]} \label{eq:map-delay-bound-1-3-1} \\
\le&\mathbf{E}\Big{[}\sum_{i,j} q_{ij}(t-K_\rho) \xi(K_\rho) \Big{]} + \sum_{i,j}\mathbf{E}\Big{[}\sum_{k=0}^{K_\rho-1}a_{ij}(t-K_\rho+k)A^\dagger_{ij}(t)\Big{]} \label{eq:map-delay-bound-1-3-2} \\
=&\xi(K_\rho)\mathbf{E}\big{[}\|Q(t-K_\rho)\|_1\big{]} + \sum_{i,j}\mathbf{E}\Big{[}\sum_{k=0}^{K_\rho-1}a_{ij}(t-K_\rho+k)A^\dagger_{ij}(t)\Big{]}.
\end{align}
\endgroup

\autoref{eq:map-delay-bound-1-1} is due to~\autoref{eq:queueing-dynamics-entrywise}. 
Inequality \autoref{eq:map-delay-bound-1-2} 
is due to $d_{ij}(t)\ge 0$ for any $i,j$ and $t$. 
Inequality~\autoref{eq:map-delay-bound-1-3-1} is due to~\autoref{fact:dependent-arrival-rate-map}. 
Inequality \autoref{eq:map-delay-bound-1-3-2} is due to~\autoref{eq:load-bound}. 


Using the weak departure inequality~\autoref{eq:QPS-property}, we have the second expectation term on the RHS of~\autoref{eq:map-delay-bound-1}, 
\[
\mathbf{E}\big{[}\sum_{i,j}q_{ij}(t)D^\dagger_{ij}(t)\big{]}=\mathbf{E}\Big{[}\mathbf{E}\big{[}\sum_{i,j}q_{ij}(t)D^\dagger_{ij}(t)\mid Q(t)\big{]}\Big{]}\ge \mathbf{E}[\|Q(t)\|_1].
\] 
Therefore, 
we have, in steady state 
\begingroup
\allowdisplaybreaks
\begin{align}
&\mathbf{E}\big{[}2{\textstyle \sum\limits_{i,j}}q_{ij}(t)\big{(}A^\dagger_{ij}(t)\!-\!D^\dagger_{ij}(t)\big{)}\big{]} \notag \\
\le&2\Big{(}\xi(K_\rho)\mathbf{E}\big{[}\|Q(t-K_\rho)\|_1\big{]}+ \sum_{i,j}\mathbf{E}\Big{[}\sum_{k=0}^{K_\rho-1}a_{ij}(t-K_\rho+k)A^\dagger_{ij}(t)\Big{]} - \mathbf{E}[\|Q(t)\|_1]\Big{)} \notag \\
\le&2\Big{(}\xi(K_\rho)\mathbf{E}\big{[}\|\bar{Q}\|_1\big{]}+ \sum_{i,j}\mathbf{E}\Big{[}\sum_{k=0}^{K_\rho-1}a_{ij}(t-K_\rho+k)A^\dagger_{ij}(t)\Big{]} - \mathbf{E}[\|\bar{Q}\|_1]\Big{)} \label{eq:steady-state-prob} \\
=&2\Big{(}\big{(}\xi(K_\rho)-1\big{)}\mathbf{E}\big{[}\|\bar{Q}\|_1\big{]} + \sum_{i,j}\mathbf{E}\Big{[}\sum_{k=0}^{K_\rho-1}a_{ij}(k-K_\rho)A^\dagger_{ij}(0)\Big{]}\Big{)}.
\end{align}
\endgroup
\autoref{eq:steady-state-prob} is due to the fact $\mathbf{E}\big{[}\|Q(\tau)\|_1\big{]}=\mathbf{E}\big{[}\|\bar{Q}\|_1\big{]}$ for any $\tau$ in steady state.

Now we prove~\autoref{eq:term-2-tbp-map-delay}. By simple calculations, we have
\begin{align}
&\mathbf{E}\big{[}\big{(}a_{ij}(t)-d_{ij}(t)\big{)}\big{(}A^\dagger_{ij}(t)-D^\dagger_{ij}(t)\big{)}\big{]}\notag \\
=&\mathbf{E}[a_{ij}(t)A^\dagger_{ij}(t) - d_{ij}(t)A^\dagger_{ij}(t)-a_{ij}(t)D^\dagger_{ij}(t)+d_{ij}(t)D^\dagger_{ij}(t)]\nonumber\\
=&\mathbf{E}[a_{ij}(t)A^\dagger_{ij}(t)] - \mathbf{E}[d_{ij}(t)A^\dagger_{ij}(t)]\notag \\
&\qquad{}-\mathbf{E}[a_{ij}(t)D^\dagger_{ij}(t)]+\mathbf{E}[d_{ij}(t)D^\dagger_{ij}(t)]\nonumber\\
=&\mathbf{E}[a_{ij}(t)A^\dagger_{ij}(t)] - \mathbf{E}[d_{ij}(t)A^\dagger_{ij}(t)]\notag \\
&\qquad{}-\mathbf{E}[a_{ij}(t)D^\dagger_{ij}(t)]+\mathbf{E}[d_{ij}(t)]\label{eq:using-dep-fact} \\
\le&\mathbf{E}[a_{ij}(t)A^\dagger_{ij}(t)]+\mathbf{E}[d_{ij}(t)]. \label{eq:using-nonnegative-fact}
\end{align}

In arriving at~\autoref{eq:using-dep-fact}, we have used \autoref{fact:departure-general-2}, {\it i.e.,} for any $i,j$, $d_{ij}(t)D^\dagger_{ij}(t)= d_{ij}(t)$. Inequality~\autoref{eq:using-nonnegative-fact} is due to $a_{ij}(t)\ge 0,d_{ij}(t)\ge 0$ for any $i,j$ at any time slot $t$. 

We note that $\{Q(t),X(t)\}_{t=0}^{\infty}$ is positive recurrent. Therefore,
we have, in steady state, for any $i,j$, $\mathbf{E}[d_{ij}(t)]\!=\!\lambda_{ij}$.
Hence, we have, in steady state,  
\begin{align}
&\mathbf{E}\big{[}\sum_{i,j}\big{(}a_{ij}(t)\!-\!d_{ij}(t)\big{)}\big{(}A^\dagger_{ij}(t)\!-\!D^\dagger_{ij}(t)\big{)}\big{]} \notag \\
=&\sum_{i,j}\mathbf{E}\big{[}\big{(}a_{ij}(t)\!-\!d_{ij}(t)\big{)}\big{(}A^\dagger_{ij}(t)\!-\!D^\dagger_{ij}(t)\big{)}\big{]} \notag \\
\le&\sum_{i,j}\left(\mathbf{E}[a_{ij}(0)A^\dagger_{ij}(0)] +\lambda_{ij}\right).\notag
\end{align}

\noindent
{\bf Part (II).} 
Now we proceed to prove (II). Since any $a_{ij}(t)$ now is independent of each other,
we have, for any $i,j$
\begingroup
\allowdisplaybreaks
\begin{align}
    &\mathbf{E}[a_{ij}(t)A^\dagger_{ij}(t)] \notag \\
    =& \mathbf{E}\left[a_{ij}(t)\left(A^\dagger_{ij}(t)-a_{ij}(t)\right) + a_{ij}^2(t)\right] \notag\\
    =&\mathbf{E}\left[a_{ij}(t)\left(A^\dagger_{ij}(t)-a_{ij}(t)\right)\right]+\mathbf{E}\left[a_{ij}^2(t)\right] \notag\\
    =&\mathbf{E}\left[a_{ij}(t)\right]\mathbf{E}\left[\sum_{(l,w)\in Q_{ij}^\dagger\setminus\{(i,j)\}}a_{lw}(t)\right]+\mathbf{E}\left[a_{ij}^2(t)\right] \notag\\
    =&\lambda_{ij}(\Lambda^\dagger_{ij}-\lambda_{ij}) + \mathbf{E}\left[a_{ij}^2(t)\right] \notag\\
    =&\sigma_{ij}^2+\lambda_{ij}\Lambda_{ij}^\dagger, \label{eq:exp-1}
\end{align}
\endgroup
where $\sigma^2_{ij}\triangleq \mathbf{E}\left[a_{ij}^2(t)\right] - \lambda^2_{ij}$ is the steady-state variance of $a_{ij}(t)$.

Similarly, we have, for any $i,j$, any $t$ and any integer $k>0$
\begin{equation}\label{eq:exp-2}
    \mathbf{E}[a_{ij}(t)A^\dagger_{ij}(t+k)]=\lambda_{ij}\Lambda_{ij}^\dagger + \theta_{ij}(k).
\end{equation}
where
$\theta_{ij}(k)\triangleq \mathbf{E}\left[a_{ij}(t)a_{ij}(t+k)\right] - \lambda^2_{ij}$ is
the auto-correlation in $a_{ij}(t)$ in steady state.

Using \autoref{eq:final-delay-bound-general}, \autoref{eq:exp-1}, and \autoref{eq:exp-2}, we
have,
\begin{align*}
\mathbf{E}[\|\bar{Q}\|_1]\!\le\!\frac{1}{2(1\!-\!\xi(K_\rho))}\left(\sum_{i,j}\left(\sigma_{ij}^2 \!+\! \lambda_{ij}\Lambda_{ij}^\dagger+\lambda_{ij}\right) \!+\! 2\sum_{i,j}\sum_{k=1}^{K_\rho}\left(\lambda_{ij}\Lambda_{ij}^\dagger \!+\!\theta_{ij}(k)\right)\right). \label{eq:final-delay-bound-general-independent}
\end{align*} 

\section{How Mean Delay of QPS-r Scales with r}\label{app:diff-r}

\begin{figure}[t]
\centering
\includegraphics[width=\textwidth]{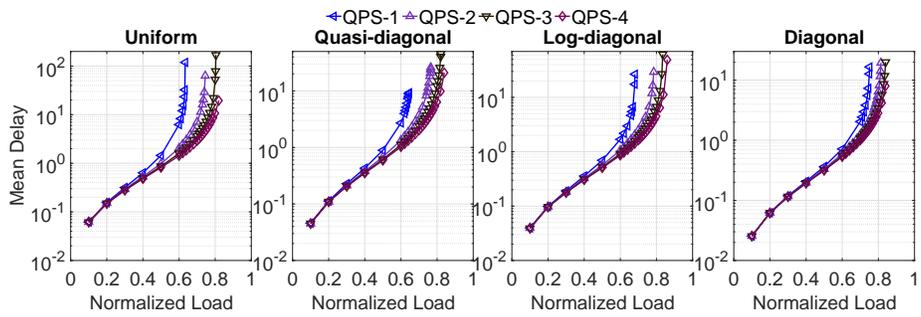}
\caption{Mean delays of QPS-r ($r=1,2,3,4$) under the $4$ traffic patterns.}\label{fig: delay-load-qpsr}
\end{figure}

\autoref{fig: delay-load-qpsr} presents the mean delay performance of QPS-r, with
$r=1,2,3,4$, as a function of the offered loads, under Bernoulli
{\it i.i.d.} arrivals with the $4$ traffic patterns.
It shows that both the maximum sustainable throughput and the mean delay performance
improve as $r$ increases. However, as $r\!>\! 3$, the improvement becomes marginal.